\documentclass[12pt]{article}
\usepackage[utf8]{inputenc}
\usepackage[T1]{fontenc}
\usepackage{amsmath, amsfonts, bm, bbm, amssymb, amsthm}
\usepackage{mathtools, mathabx}
\usepackage{circuitikz}
\usetikzlibrary{decorations.pathreplacing}
\usepackage{fullpage}
\usepackage{enumerate}
\usepackage[open, openlevel = 4]{bookmark}
\usepackage{hyperref,url,doi}
\usepackage[all]{nowidow}
\usepackage{mdframed, float, multirow, adjustbox}
\usepackage{listings}
\usepackage{xcolor}
\usepackage{diagbox}
\usepackage{algorithm}
\usepackage{algpseudocode}
\usepackage{diagbox, adjustbox}
\usepackage{booktabs}
\usepackage{siunitx}
\usepackage{subcaption}
\usepackage{appendix}
\usepackage[authoryear, sort&compress]{natbib}
\usepackage{enumitem}
\usepackage{pgfplots}
\pgfplotsset{compat=1.15}
\usetikzlibrary{3d,arrows.meta}

\numberwithin{equation}{section}
\numberwithin{figure}{section}
\numberwithin{table}{section}

\theoremstyle{plain}
\newtheorem{theorem}{Theorem}[section]
\newtheorem{lemma}{Lemma}[section]
\newtheorem{corollary}{Corollary}[section]
\newtheorem{proposition}{Proposition}[section]

\theoremstyle{definition}
\newtheorem{definition}{Definition}[section]

\newtheorem{assumption}{}[section]

\newtheorem{remark}{Remark}[section]

\hypersetup{
	colorlinks = true,
	citecolor = blue,
	linkcolor = red,
}

\newcommand{\pr}[1]{\mathbb{P} \big( #1 \big)}

\newcommand{\spt}[1]{\mathrm{supp}(#1)}

\newcommand{\Real}{\mathbb{R}}
\newcommand{\expect}[1]{\mathbb{E} \left[ #1 \right]}

\newcommand{\leb}{\text{leb}_{\vert \lbrack 0, 1 \rbrack}}
\newcommand{\C}{\mathcal{C}}
\newcommand{\supquant}[1]{\sup_{#1} \lvert F_n^{-1}(u) - \hat{F}_n^{-1}(u) \rvert}

\title{An Optimal Transportation Approach for Improved Confidence Intervals}

\author{Christophe Quentin Valvason$^\dagger$ 
	{\small and  }
	Eustasio del Barrio$^\#$
	{\small and  }
	Stefan Sperlich$^\dagger$ \\[2mm]
	{\small $\dagger$ Geneva School of Economics and Management, University of Geneva,} \\
	{ \small	Bd du Pont d'Arve 40, 1211 Geneva 4, Switzerland}	 
	\footnote{The first and third author acknowledge financial support from
		the project "Uniform- and Post-selection inference for Mixed Parameters", 200021-192345 of the Swiss National Science Foundation.}
	\\[2mm]
	{\small $\#$ IMUVA, Universidad de Valladolid, Valladolid, Spain}	
}

\begin{document}
	\maketitle
	\begin{abstract}
		\noindent
		Optimal transport methods have recently attracted a lot of attention in statistics. Their appeal lies in providing a geometric framework for comparing probability measures, leading to new perspectives on classical problems.	
		A central problem in statistics is the construction of valid confidence sets as fundamental inferential tools in practice. A well-known problem is that for complex problems or relatively small samples, their asymptotic approximations often show poor performance. This suggests to apply optimal transport methods when constructing confidence sets for hard problems to improve their coverage properties. We introduce such a procedure, derive the theoretical framework studying consistency and error bounds for the coverage probability of the resulting intervals. To guarantee feasibility in practice, we propose data-driven choices for our hyper parameters.
		This approach extends classical quantile-based confidence intervals by leveraging optimal couplings to minimize coverage deviations. Simulations demonstrate striking performance in different estimation problems, outperforming standard methods in accuracy and robustness. 
		\\[2mm]	
	\end{abstract}
	
	\section{Introduction}\label{sec:intro}
	
	
	A central problem in statistical inference is the construction of reliable confidence intervals for a parameter of interest. The classical approach on the real line proceeds through quantile inversion of the estimators' distribution. In practice, the exact finite-sample distribution is unknown such that one has to resort to asymptotic or Monte Carlo approximations. For complex problems or small samples this may give intervals that neither have coverage close to the prescribed nominal level nor the optimal length.
	
	Recently, Optimal Transport (henceforth OT) -based methods have attracted a lot of attention in statistics, e.g.\ for the definition of multivariate quantiles \citep{hallin2021distributions,delbarrio2024nonparametric} 
	and robust estimation \citep{ma2025inference}. 
	The appeal lies in its geometric framework it provides, which offers new perspectives on classical problems. In this spirit, we construct confidence intervals from a known source distribution pushed to the actual inference problem that is too hard for a classical approach.  
	
	Let $\leb$ denote the Lebesgue measure restricted to the unit interval, and let $P$ be the probability measure of interest. On the real line and for any non-negative strictly convex cost function $h(x - y)$, the optimal transport map pushing $\leb$ to $P$ is given by the monotone rearrangement \citep[Theorem~2.9]{santambrogio2015optimal}
	\begin{equation}\label{eq:1-optimal_map_uniform}
		T(u) \coloneqq F_P^{-1}(u), \quad \text{for } \leb \text{-almost all } 0 \varleq u \varleq 1,
	\end{equation}
	where $F_P(t) \coloneqq P\big((-\infty, t]\big)$ denotes the cumulative distribution function (c.d.f.) of $P$.
	
	Let $X_1,\ldots, X_n$ be i.i.d.\ random variables, and $\hat{\theta}_n(X_1, \ldots, X_n) = \hat{\theta}_n$ an unbiased estimator of a parameter $\theta_0 \in \Real$. Define $P_n$ as the probability distribution of 
	\[	   \vartheta_n \coloneqq r_n(\hat{\theta}_n - \theta_0) ,   \]
	where $r_n$ is a deterministic sequence which is allowed to go to infinity.
	If $P_n$ were known exactly and absolutely continuous with respect to the Lebesgue measure, a two-sided confidence interval could be obtained by quantile inversion,
	\begin{equation} \label{eq-CI-quantinv}
		\mathbb{P}\Big(
		\hat{\theta}_n - r_n^{-1} F_n^{-1}(1 - \alpha/2) \varleq 
		\theta_0 \varleq \hat{\theta}_n - r_n^{-1} F_n^{-1}(\alpha/2)
		\Big) = 1 - \alpha ,
	\end{equation}
	where $\alpha \in (0,1)$ denotes the prescribed nominal level, and $F_n$ the c.d.f.\ of $P_n$.
	Thus, the construction of valid confidence intervals reduces to the knowledge of the quantile function or, equivalently, to the  OT map pushing $\leb$ to $P_n$.
	
	A natural extension of the above construction is obtained by allowing for other source distributions than $\leb$. Let $Q$ be any probability measure absolutely continuous with respect to Lebesgue measure and with c.d.f.\ $F_Q$.
	Then there exists a unique  OT map pushing $Q$ to $P_n$ given by
	$	T_{Q \to P_n}(x) \coloneqq F_{n}^{-1}\big( F_Q(x) \big) , \ $ 
		\text{for $Q$-almost all }$x \in \spt{Q}$,  
   where $F_{n}^{-1}$ is the generalized inverse of $F_{n}$ if $P_n$ gives mass to atoms. 
	For brevity, since the target distribution $P_n$ will remain fixed throughout, we write $T_Q = T_{Q \to P_n}$. 
	
	Define the central quantile set under $Q$ as 
	\[
	\C = [ F_Q^{-1}(\alpha / 2), F_Q^{-1}(1 - \alpha/2) ].
	\]
	A $(1 - \alpha)$ confidence interval is obtained as
	\begin{equation}\label{eq:1-confidence_interval_population_ac}
		\mathcal{I}_{Q} \coloneqq 
		\Big[  \hat{\theta}_n - r_n^{-1}T_Q\big(F_Q^{-1}(1 - \alpha/2)\big),
		\hat{\theta}_n - r_n^{-1}T_Q\big(F_Q^{-1}(\alpha / 2)\big)		\Big]
	\end{equation}
with
	$	\pr{\mathcal{I}_Q \ni \theta_0} =
		P_n\big(T_Q(\C)\big) \vargeq
		Q\big(\C\big) = 1 - \alpha ,\
	$
 	 following from the monotonicity of $T_Q$.
	\begin{remark}
		In principle, any admissible set $\C$ satisfying $Q(\C) \vargeq 1 - \alpha$ can be used to construct a confidence interval. 
		In practice, the symmetric choice 
		\[
		\C = [F_Q^{-1}(\alpha/2),\, F_Q^{-1}(1 - \alpha/2)],
		\]
		is adopted mainly for convenience and interpretability.
	\end{remark}
	
	On the real line, the set of absolutely continuous measures is compatible with the Lebesgue measure~\citep[Definition~2.3.1]{panaretos2020invitation}. For $Q = \leb$ the  OT map is the quantile function of $P_n$. 
	Then the 'original' confidence interval is recovered, and its property depends only on the estimated or assumed $P_n$.
	
	While this classical quantile-based interval is exact if $P_n$ is known, it is typically not the case
	in finite samples, and approximations typically used may lead to coverage deviations. 
	This motivates considering alternative $Q$ as
	under mild assumptions on $P_n$, a source measure supported on a finite set of atoms allows the construction of confidence intervals with coverage probabilities closer to the nominal level, especially for small $n$.

The use of discrete source distributions is motivated by statistical and computational considerations.
From a statistical viewpoint, discretization acts as a regularization mechanism for an intrinsically ill-posed problem.
In small sample settings, classical confidence intervals from the empirical distributions can easily exhibit incorrect coverage.
Constraining the source measure to have finitely many atoms stabilizes the construction and yields confidence intervals with controlled properties.
	
	\begin{figure}[htb]
		\centering \vspace{-0.25cm}
		\includegraphics[width=0.75\linewidth,height=4.5cm]{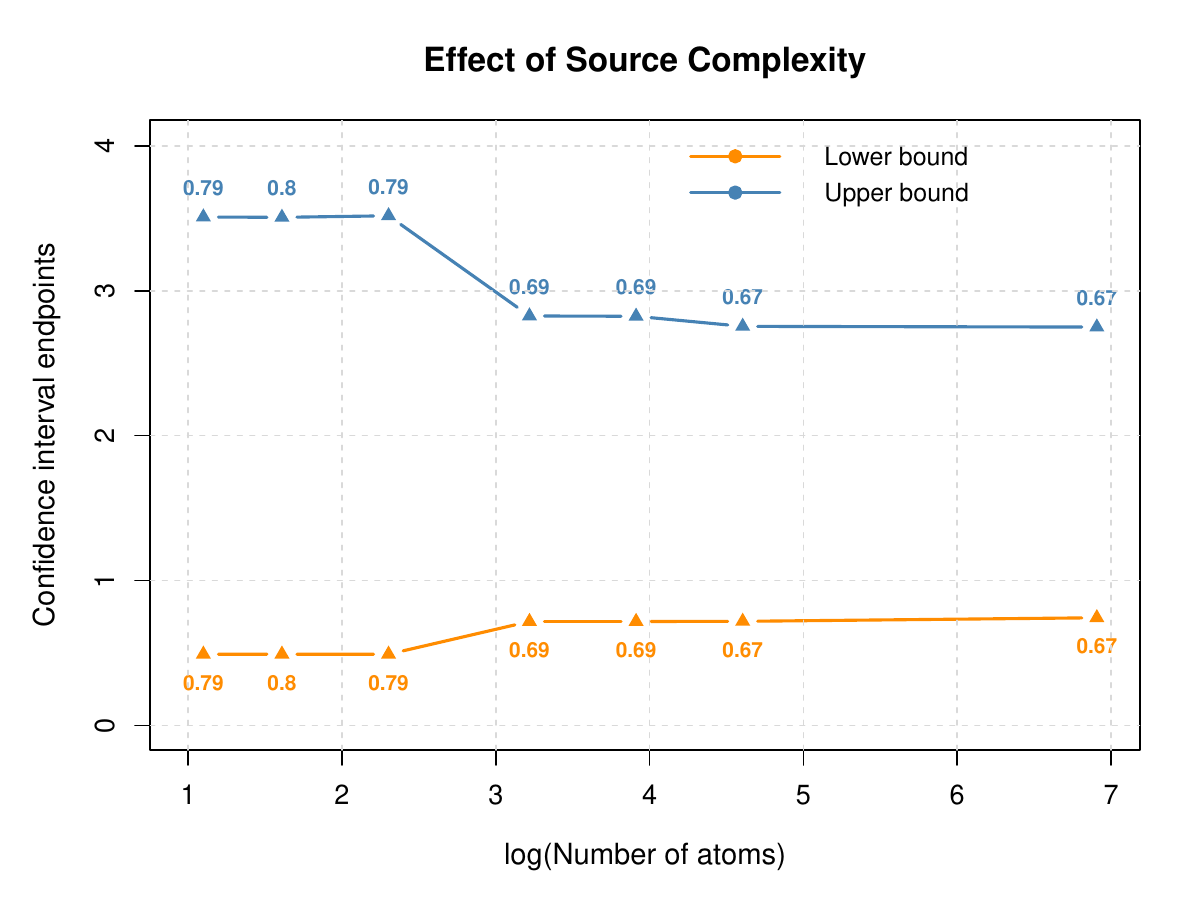}
		\vspace{-0.5cm}
		\caption{Average  OT-based confidence interval endpoints (for the mean of log-normal distributed $X_i$ with $n = 5$) as a function of $M$ of support points of $Q$. Increasing $M$ reduces the regularization induced by $Q$.}
		\label{fig:M_complexity}
	\end{figure}

	Figure~\ref{fig:M_complexity} illustrates the effect of the number of support points $M$. 
	For small $M$, the  OT-based interval restricts the rearrangements of probability mass, which induces a strong regularization effect. 
	As $M$ increases, $Q$ becomes more flexible and the interval endpoints progressively stabilize. 
	This transition is particularly visible when the horizontal axis is represented on a logarithmic scale: beyond a moderate value of $\log(M)$, further increases in the number of atoms have little impact on the interval bounds. 

\begin{figure}[htb]
	\centering \vspace{-0.25cm}
	\includegraphics[width=0.85\linewidth,height=7.75cm]{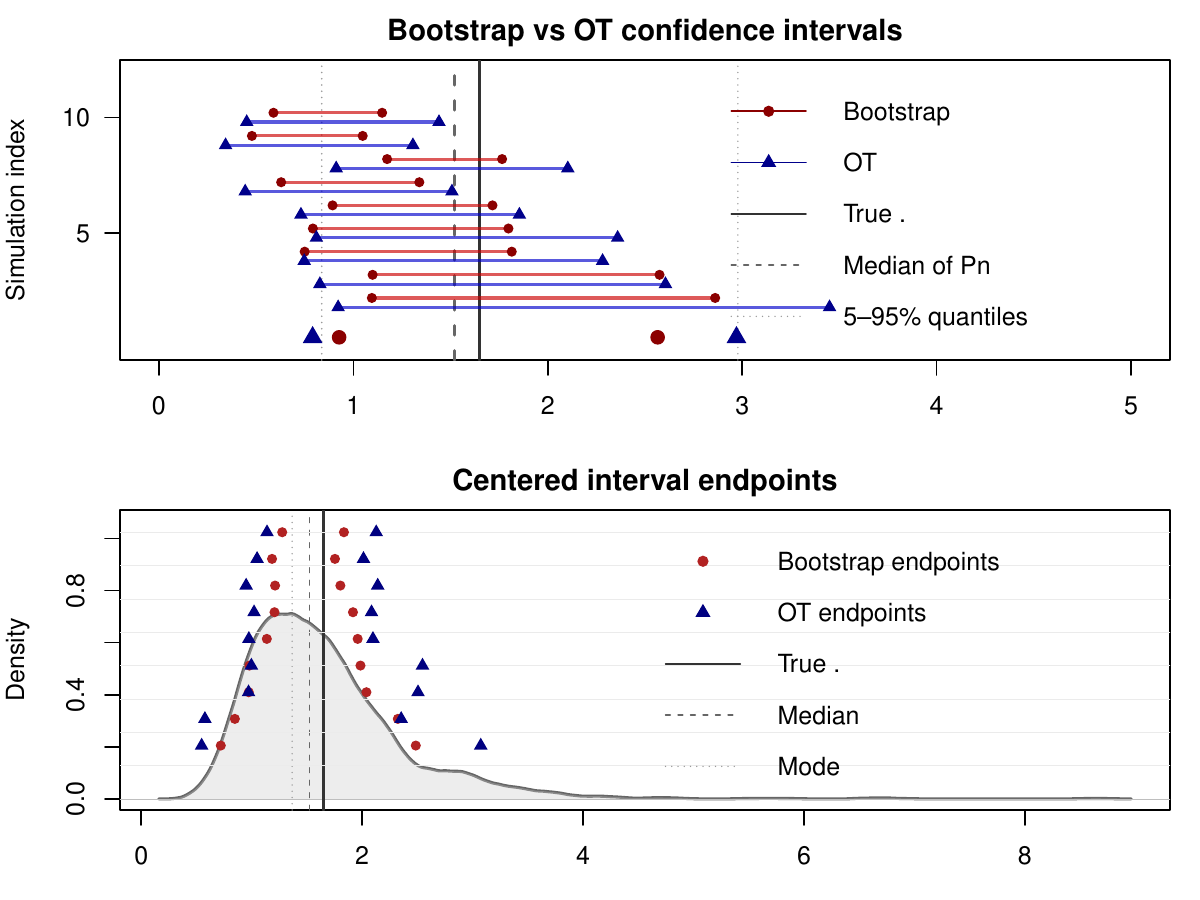}
	\vspace{-0.5cm}
	\caption{Bootstrap and OT based confidence intervals for the mean of a log-normal distribution from independent samples of $n=5$; (top) the triangles (OT) and circles (bootstrap) without segment represent the mean endpoints over 500 confidence intervals, (bottom) intervals centered at the median. 
		}\label{fig:OT_vs_Boot_regularization}
		\vspace{-0.15cm}
\end{figure}

	OT plays a central role here as
	it provides the mathematical tools that connect discrete regularization with sampling distribution while preserving exact mass constraints being crucial for obtaining finite-sample coverage guarantees and non-asymptotic bounds.
	Figure~\ref{fig:OT_vs_Boot_regularization} illustrates this by comparing bootstrap and  OT-based intervals across repeated samples.
	We simulated 500 small samples ($n=5$) from a log-normal distribution and estimated the mean. The number of atoms $M=10$ of $Q$ is fixed, but their position may change.
	The top panel shows nine of these intervals plus the average interval bounds of all 500, while the bottom panel displays their bounds (centered around the median) relative to the density of the sampling distribution. 
	The average of the interval bounds of the  OT based intervals meet almost exactly the desired quantiles. Moreover, these intervals adapt to the asymmetric geometry of the sampling distribution. 
	This reflects the geometric nature of the  OT framework: the intervals are obtained through global mass rearrangements between distributions rather than through local approximations.  
	Therefore they respect the geometry of the underlying distribution adapting to skewness and tail behavior, while stabilizing the resulting interval endpoints.
	This is closely related to the theory of statistical inverse problems, where discretization constraints are commonly used to regularize ill-posed estimation tasks
	\citep{kaipio2005statistical}.
	Optimal transport provides the missing link that enables a mathematically precise treatment of these ideas in a finite-sample setting.

	
	When $Q$ is purely atomic and the target distribution $P_n$ is absolutely continuous, the corresponding transport problem falls into the class of \emph{semi-discrete optimal transport} \citep[Ch.~5]{peyre2020computational}. 
	In this work we concentrate on the problem where $P_n$ is replaced by an empirical or bootstrap distribution $\hat P_n$, so that both measures are discrete. 
	Consequently, the relevant object is a discrete--discrete optimal coupling, and 
	the problem reduces to a linear program whose complexity depends on the support of $Q$.
	
		
\paragraph{Related Works}
Recent advances in  OT have inspired profound developments in multivariate quantile theory. In particular, \citet{chernozhukov2017monge} introduced a center-outward quantile function derived via  OT, enabling multivariate generalizations of ranks and signs. The work of \citet{hallin2021distributions} extended this approach by formalizing multivariate distribution and quantile functions based on  OT, establishing Glivenko–Cantelli–type uniform consistency of empirical rank maps. Further, \citet{ghosal2022multivariate} continued in this line by analyzing the rates of convergence of empirical  OT-based quantiles and ranks, and by leveraging them for some nonparametric testing problems. Similarly, \citet{delbarrio2024nonparametric} developed nonparametric framework for quantile regression that accommodates multivariate response variables by leveraging center-outward quantile functions.

More broadly,  OT has been used in statistical inference through Wasserstein distances by various authors; see \citep{delbarrio1999central,delbarrio1999tests,fournier2015rate,sommerfeld2017inference,weed2019sharp, zhang2025distributionally}, and \citep{panaretos2019statistical} for a recent survey. Other recent contributions to the theory of empirical  OT include minimax rates of estimation and plug-in estimators of smooth  OT maps \citep{hutter2021minimax,manole2024plugin,ponnoprat2025minimax} as well as uniform confidence bands in dimension one \citep{ponnoprat2024uniform}.  
	
Our method extends classical confidence interval estimation by leveraging an  OT framework to optimize the procedure coverage. 
We provide rigorous finite-sample guarantees, showing that the resulting intervals achieve coverage probability close to the nominal level. 
In this sense, our work is not a contribution to  OT per se 
but an application of this methodology to a specific but highly relevant problem in statistical practices.
Therefore it differs substantially in aim and methodology from the existing literature on  OT. 
Our objective is to obtain significant improvements in coverage and efficiency for hard problems, especially when data a sparse.
	
	
	Although our primary focus is on finite-sample guarantees, we first establish large-sample results. In particular, Theorem~\ref{th:3-hausdorff_consistency} shows that, under mild conditions, the empirical confidence interval constructed via  OT converges in probability to its population counterpart.
%
	Then we derive explicit finite-sample lower bounds for the coverage probability, highlighting the dependence on the choice of $Q$. 
Under mild assumptions, appropriate selection of $Q$ yields coverage errors that decay exponentially with the sample size. Additionally, we provide a finite-sample analysis of the confidence interval length, relating the empirical length to its population analogue.
	
	Based on the theoretical results, we propose practical procedures to select the source distribution from observed data. We focus on discretized Beta distributions and introduce a data-driven method to select both the parameters and the discretization, optimizing the finite-sample performance of the resulting confidence intervals.
%
We illustrate the finite-sample performance of the proposed method through extensive Monte Carlo simulations and an application. Coverage probabilities, including averages, variations, and mean squared errors, are evaluated for two representative estimation problems across varying sample sizes. Our method consistently demonstrates improved coverage accuracy and stability, providing insight into the behavior predicted by the theory.
	
	The remainder of the paper is organized as follows. In the next section, we introduce the problem, notation, and key facts from  OT theory. Section~\ref{sec:ECIbOT} presents our construction of confidence intervals.
Section~\ref{sec:miscel} provides analytical results on interval length and proposes data-driven selection methods for all hyperparameters. Section~\ref{sec:montecarlo} presents simulation studies and an application for illustrating the practical benefits of our approach. Section~\ref{sec:concl} concludes with a discussion.

	\section{The Statistical Problem and Optimal Transport} \label{sec:optimaltransport}
	
		\subsection{Optimal Transport Revisited}
	
		Before presenting the main results, we review several fundamental facts from  OT theory.
		For a comprehensive treatment of the theory, we refer to the seminal monographs~\citep{villani2003topics,villani2009old,panaretos2020invitation} 
		and \citep{santambrogio2015optimal}.
		
		We consider an OT problem on $\Real$ with quadratic cost functions $c(x,y) = \lvert x - y \rvert^2 / 2$.
		Let $\mathcal{P}(\Real)$ denote the set of Borel probability measures on $\Real$, and $\mathcal{B}(\Real)$ the Borel $\sigma$-algebra. 
		Denote by $\mathcal{P}_2(\Real) \subset \mathcal{P(\Real)}$ the set of probability measures with finite second moment
		\[
		\mathcal{P}_2(\Real) \coloneqq \{\mu \in \mathcal{P}(\Real) : \int_\Real \lvert x \rvert^2 d\mu(x) < \infty \}
		\]
		and consider two such measures $P,Q \in \mathcal{P}_2(\Real)$.
		The so-called Monge's problem 
		seeks a Borel measurable map $T: \mathcal{\Real} \rightarrow \mathcal{\Real}$ minimizing the transportation problem	
		\begin{equation}\label{eq:02-monge_total_cost}
				\inf_{T: T_\#P = Q} M [ T ] \coloneqq \frac{1}{2}\int_{\Real} \lvert x - T(x) \rvert^2 dP(x),
			\end{equation}
		where $T_{\#}P$ denotes the push-forward of the probability measure $P$ under the mapping $T$,
		\[ T_{\#}P(A) = P\left(T^{-1}(A)\right) , \qquad	\text{for all } A \in \mathcal{B}(\Real).\]
		In general, minimizing \eqref{eq:02-monge_total_cost} is difficult or even impossible, i.e., the admissible set of maps may be empty.
		This can be solved by allowing the transported mass to be split among different locations; this is the approach of \citep{kantorovitch1942transportation}. He sought a \emph{transference plan} $\gamma$, i.e.\ a joint probability measure on $\big(\Real, \mathcal{B}(\Real)\big) \times \big(\Real, \mathcal{B}(\Real)\big)$ such that its marginals are $P$ and $Q$.
		Then the OT problem becomes
		\begin{equation}\label{eq:03-kantorovich_total_cost}
				\inf_{\gamma \in \Gamma(P,Q)} K [ \gamma ] \coloneqq  \frac{1}{2}\int_{\Real \times \Real} \lvert x - y \rvert^2 d\gamma(x,y)  .	
			\end{equation}
		where $\Gamma(P,Q)$ is the set of joint probability measures with marginals $P, Q \in \mathcal{P}_2(\Real)$. This set is never empty as it always contains the product probability measure $P \otimes Q$. Unlike to Monge's problem, this admits a general existence result \citep[Theorem~4.1]{villani2009old}.
			
		The functional $K [ \gamma ]$ is linear in $\gamma$ and the marginal constraints are convex. Consequently, the Kantorovich's problem is an infinite-dimensional convex linear programming problem that admits a dual formulation,
		\begin{eqnarray}\label{eq:03-dual_functional}
			& \max_{(\varphi, \psi) \in \Phi_c} J [ \varphi, \psi ] \coloneqq
				\int_{\Real} \varphi(x) dP(x) + \int_{\Real} \psi(y) dQ(y) 
				\ , \ \mbox{for} &  
		\\ &
		\Phi_c := \Big\{(\varphi, \psi) \in C_b(\Real) \times C_b(\Real) : \varphi(x) + \psi(y) \varleq \frac{1}{2} \lvert x - y \rvert^2\Big\}  . & \nonumber
		\end{eqnarray}
		It follows from results in convex analysis that if $\varphi$ is optimal for the dual problem, it is lower-semicontinuous and convex.
		Moreover, $\psi$ can be fixed as the Legendre transform $\varphi^{\ast}$ of $\varphi$ \citep[Chapter~2.]{villani2003topics} giving the optimal pair $(\varphi, \varphi^{\ast})$.
		
		\begin{definition}[Cyclical Monotonicity]
				A subset $S \subset \Real \times \Real$ is said to be \emph{cyclically monotone}, if for all $k \in \mathbb{N}$ and for any finite family $(x_1, y_1), \ldots, (x_k, y_k)$ of points in $S$ 
				\begin{equation}\label{eq:03-cyclical_monotonicity}
						\sum_{i = 1}^{k}\frac{\lvert x_i - y_i \rvert^2}{2} \varleq  
						\sum_{i = 1}^{k} \frac{\lvert x_i - y_{i + 1} \rvert^2}{2},
					\end{equation}
				holds with the convention $y_{k+1} = y_1$.
				In dimension one this reduces to plain monotonicity.
			\end{definition}
		By \citep[Theorem~24.8]{rockafellar1970convex}, a set $S$ is cyclically monotone set iff it is contained in the graph of the subdifferential of a lower semicontinuous convex function $f \not \equiv + \infty$, i.e.\
		$
		S \subset \mathrm{Graph}(\partial f) = \bigcup_{x \in \Real} \{x\} \times \partial f(x),
		$
		where the subdifferential at $x \in \Real$ is defined as
		\begin{equation} \label{eq:partialf}
				\partial f(x) \coloneqq
				\Big\{y \in \Real : f(z) \vargeq f(x) + y(z - x), \quad \forall z \in \Real\Big\} .
			\end{equation}
		If $f$ is differentiable at $x_0 \in \Real$, the subdifferential reduces to the singleton $\partial f(x_0) = \{f'(x_0)\}$.
		It follows that if $\gamma \in \Gamma(P,Q)$ is an optimal coupling, then its support is contained in a cyclically monotone set, which is contained in the subdifferential of a lower semicontinuous convex potential $\varphi$,	i.e.
		$ 
		y \in \partial \varphi(x) $ \text{for $\gamma$-almost all } $(x,y) \in \Real\times\Real .$ 
	
	\subsection{Optimal Transport for Confidence Sets and Intervals}
	
	For the sake of presentation, we concentrate on two-sided confidence intervals with wanted coverage probability of $(1-\alpha)$, $\alpha \in (0,1)$.
	Procedure and theoretical results remain valid for one-sided confidence 
	intervals under some straightforward modifications.
	
	Let $Q \in \mathcal{P}_2(\Real)$ be a purely atomic probability measure with finite support $\spt{Q} = \{x_1, \ldots, x_M\}$, such that for the Dirac measure $\delta_{x_i}$
	\[	
	Q = \sum_{i = 1}^{M} q_i \delta_{x_i} ,  
	\quad 
	q_i > 0 \text{ and } \sum_{i = 1}^{M} q_i = 1     
	\]
	and $x_1 < x_2 < \ldots < x_M$ for simplicity but without loss of generality.
	For constructing a confidence interval, fix a finite admissible set $\C \subset \spt{Q}$ such that $Q(\C) \vargeq 1 - \alpha$.
	
	For cost function $c(x,y) = \lvert x - y \rvert^2 / 2$ consider the optimal coupling $\gamma \in \Gamma(Q,P_n)$ with a convex potential $\varphi_n$ such that the support of $\gamma$ is contained in the graph of the subdifferential $\partial \varphi_n$.
	To obtain a confidence interval in this setting, define the closed interval
	\begin{equation}\label{eq:1-interval_population_geometric}
		\mathcal{J}_{n,Q} \coloneqq 
		\operatorname{conv} \Big\{y \in \Real : \exists x \in \C \text{ with } y \in \partial\varphi_n(x)\Big\},
	\end{equation}
	where $\operatorname{conv}(A)$ denotes the convex hull of $A$. 
	Equivalently, the interval reduces to
	\begin{equation}\label{eq:1-interval_population_analytic}
		\mathcal{J}_{n,Q} = [ D^{-}\varphi_n(\min \C), D^{+}\varphi_n(\max \C) ]  ,
	\end{equation}
	since for convex potential $\varphi_n$, the left and right derivative $D^{-}\varphi_n, D^{+}\varphi_n$ exists everywhere on the interior of its domain, and the subdifferentials satisfy \citep[p.~216]{rockafellar1970convex} 
	\[
	\partial \varphi_n(x) = [ D^{-} \varphi_n(x), D^{+}\varphi_n(x) ].
	\]
	Then, a confidence interval for parameter $\theta_0$ given an unbiased estimator $\hat\theta_n$ (with notation and definitions of Section \ref{sec:intro}) is obtained by inversion, i.e.\
	\begin{equation}\label{eq:1-confidence_interval_population}
		\mathcal{I}_{n,Q} \coloneqq [
		\hat{\theta}_n - r_n^{-1} D^{+}\varphi_n(\max \C),
		\hat{\theta}_n - r_n^{-1} D^{-}\varphi_n(\min \C)	] ,
	\end{equation}
	satisfying as above 
	$
	\pr{\mathcal{I}_{n,Q} \ni \theta_0} =
	P_n\big(\mathcal{J}_{n,Q}\big) \vargeq
	Q\big(\C\big) \vargeq 1 - \alpha .	
	$
	
	\begin{remark}
		In the discrete setting, the convex potential $\varphi_n$ is not unique. For definiteness, you can fix one such convex potential. Then the resulting confidence interval is well defined and the choice of potential does not affect the validity of the coverage guarantees, since all optimal potentials produce intervals contained within the same convex hull of the transported atoms of $Q$.
	\end{remark}
	
	
	\begin{figure}[htb]
		\centering  \vspace{-0.25cm}
		\begin{tikzpicture}[>=latex, scale=0.75]
			
			\tikzset{
				bar/.style={line width=1.2pt, draw=black, fill=gray!30},           
				barC/.style={line width=1.5pt, draw=blue!70!black, fill=blue!50},
				dens/.style={smooth,  thick, black},
				arrowC/.style={blue!70!black, dashed, ->, thick},
				label/.style={font=\small},
				brace/.style={decorate,decoration={brace,amplitude=4pt}}
			}
			
			\draw[->, thick] (0,0) -- (16.0,0);
			
			
			\coordinate (a1) at (0.5,0);
			\coordinate (a2) at (1.6,0);
			\coordinate (a3) at (2.8,0); 
			\coordinate (a4) at (4.0,0); 
			\coordinate (a5) at (5.1,0); 
			\coordinate (a6) at (6.4,0); 
			\coordinate (a7) at (7.6,0);
			
			\def\h{2.0}
			\draw[bar] (a1) rectangle ++(0.15,\h*0.6);
			\draw[bar] (a2) rectangle ++(0.15,\h*0.3);
			\draw[barC] (a3) rectangle ++(0.15,\h*0.9);
			\draw[barC] (a4) rectangle ++(0.15,\h*1.3);
			\draw[barC] (a5) rectangle ++(0.15,\h*0.8);
			\draw[barC] (a6) rectangle ++(0.15,\h*1.1);
			\draw[bar] (a7) rectangle ++(0.15,\h*0.3);
			
			\node[label] at (4,3.2) {$Q$};
			
			\node[below] at (a1) {$x_1$};
			\node[below] at (a2) {$x_2$};
			\node[below] at (a3) {$x_3$};
			\node[below] at (a4) {$x_4$};
			\node[below] at (a5) {$x_5$};
			\node[below] at (a6) {$x_6$};
			\node[below] at (a7) {$x_7$};
			
				\draw[decorate, decoration={brace,amplitude=4pt,mirror}]  
			(a3.south |- 0,-0.6) -- (a6.south |- 0,-0.6) 
			node[midway,below=6pt] {$\mathcal C$};
			
			\draw[->, thick] (8.0,0.0) -- (8.0,3.8);
			
			
			\def\IL{10.2}
			\def\IR{13.6}
			
			\draw[dens, samples=200, domain=9:15.5] 
			plot (\x, {5*(0.5*exp(-((\x-12)^2)/2) + 0.2*exp(-((\x-11.0)^2)/0.5))});
			
			\begin{scope}
				\clip (\IL,0) rectangle (\IR,5);
				\draw[fill=orange!70, opacity=0.5, draw=none, samples=200, domain=9:15.5] 
				plot (\x, {5*(0.5*exp(-((\x-12)^2)/2) + 0.2*exp(-((\x-11.0)^2)/0.5))});
			\end{scope}
			
			\draw[decorate, decoration={brace,amplitude=4pt,mirror}] 
			(\IL, -0.6) -- (\IR, -0.6) 
			node[midway, below=6pt] {$\mathcal I_Q$};
			
			\node[label] at (11.4,3.2) {$P_n$};
			
			\draw[arrowC, bend left=35] (a3.north) to ({\IL},0);
			\draw[arrowC, bend left=50] (a4.north) to ({(\IL+\IR)/2},0);
			\draw[arrowC, bend left=45] (a5.north) to ({(\IL+\IR)/1.9},0);
			\draw[arrowC, bend left=35] (a6.north) to ({\IR},0);
			
		\end{tikzpicture}
		\vspace{-0.5cm}
		\caption{
			\textbf{Left:} Barplot of $Q$ with the admissible set $\mathcal C$ in blue. 
			Dashed arrows indicate the contribution of atoms in $\mathcal C$ to the confidence interval $\mathcal I_{n,Q}$. 
			\textbf{Right:} Density of $P_n$ with the orange-shaded region indicating the confidence interval $\mathcal I_{n,Q}$.
		}
	\end{figure}
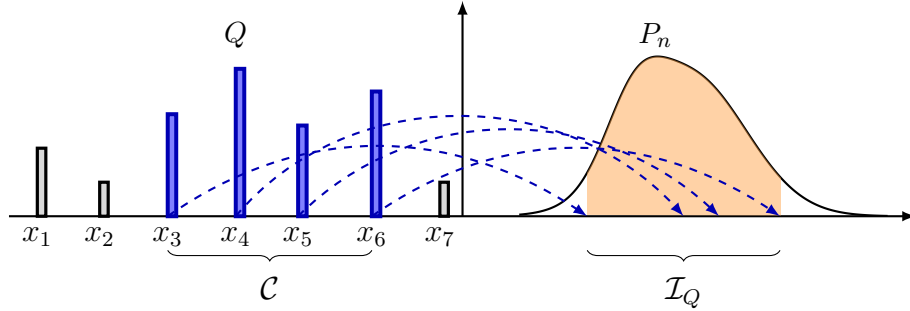
	
	Let $\pi_n(Q) \coloneqq \pr{\mathcal{I}_{n,Q} \ni \theta_0}$ be the coverage probability.
	Our formulation naturally leads to an optimization problem for minimizing the error of coverage from the nominal level, 
	\[
	Q^{\star} \in \arg\inf_{Q \in \mathcal{Q}} \phi\big(\pi_n(Q) - (1 - \alpha)\big),
	\]
	with $\phi$ a suitably chosen function to guide the optimization, and $\mathcal{Q} \subset \mathcal{P}_2(\Real)$ a class of purely atomic probability measures; see Section~\ref{sec:ECIbOT}.
	
	\begin{remark}
		Our methodology can be extended to an estimator $\hat{\theta}_n$ with bias $B_n$ by 
		defining a modified error 
		\[	\vartheta_n^{\text{(bias)}} \coloneqq 	r_n \big(\hat{\theta}_n - \theta_0 - B_n \big) ,  \]
		and repeat the construction with $\vartheta_n^{\text{(bias)}}$ replacing $\vartheta_n$.
		Note that we are not assuming $\hat\theta_n$ to be a consistent estimator, as we are interested in its finite-sample distribution $P_n$.
	\end{remark}
	
	Recall that for improving the coverage probability, we consider discrete source distributions with a finite number of atoms $M$. The OT problem may admit multiple solutions giving confidence intervals being defined through subdifferentials of a convex potential for the dual problem.
	In contrast to an absolutely continuous source distribution, the analysis of the coverage probability depends now on the choice of $Q$.
	To the best of our knowledge, we are the first who consider a confidence interval constructions under an  OT approach from a purely atomic source.

	\section{Empirical Confidence Intervals by Optimal Transport}\label{sec:ECIbOT}
	
	In practice, the probability measure $P_n$ is unknown.	
	Therefore, one can either estimates it nonparametrically or adopt a parametric model. 
	For the ease of notation, we henceforth write $\hat{P}_n$ to indicate its estimate, with the precise meaning becoming clear from context.
	Naturally, if one selects a parametric model for $P_n$, one must accept the model as valid, and any inference is performed then within this framework.
	
	We are provided by a sample of random variables $	X_i: \Omega \to \mathcal{X}_i,\ i = 1, \ldots, n, $
	from a complete probability space  $\big(\Omega, \mathcal{F}, \mathbb{P}\big)$ 
	taking value in measurable spaces $\big(\mathcal{X}_i, \mathcal{A}_i\big)$.
	For the population parameter $\theta_0 \in \Real$ of interest we have a Borel measurable estimator $\hat{\theta}_n: \mathcal{X}_1 \times \ldots \times \mathcal{X}_n \rightarrow \Real$,	and denote as before our weighted error as
	$ 
	\vartheta_n \coloneqq r_n(\hat{\theta}_n - \theta_0) $ with distribution $ P_n \in \mathcal{P}_2(\Real) .
	$ 
	%
	%
	Construct a confidence interval for $\theta_0$ from $\hat{P}_n$ via the OT based inversion as described in Section \ref{sec:optimaltransport}. 
	Having used the quadratic cost function, the Kantorovich potentials 
	$ (\hat{\varphi}_n, \hat{\varphi}_n^{\ast})$ of any optimal coupling $\hat{\gamma}_n \in \Gamma(Q, \hat{P}_n)$ 
	are convex and lower semicontinuous.
	The \emph{approximate} confidence interval for $\theta_0$ is given by 
	\begin{equation}\label{eq:3-confidence_interval_sample}
		\hat{\mathcal{I}}_{n,Q} \coloneqq 
		\Big[ 
		\hat{\theta}_n - r_n^{-1} D^{+}\hat{\varphi}_n(c_{b}) ,\
		\hat{\theta}_n - r_n^{-1} D^{-}\hat{\varphi}_n(c_{a})
		\Big] ,
	\end{equation}
	where $D^{\pm}\hat{\varphi}_n$ denote the left and right derivative of $\hat{\varphi}_n$, and $c_a = \min \C, \ c_b = \max \C \ .$
	
	\begin{remark}
		The proposed methodology can be extended to more general lower semicontinuous cost function $c(x,y)$.
		In the generalized setting, convex potential are replaced by $c$-concave potentials, cyclically monotone sets by $c$-cyclically monotone sets, and the subdifferential by the $c$-superdifferential, with analogous modifications; see Appendix~\ref{app:convex_cost_confidence_interval}.
	\end{remark}
	
	For a nonparametric estimator of $\hat P_n$ it may be necessary to generate replicates
	\[
	\hat{\vartheta}_n^{(1) \ast}, \ldots, \hat{\vartheta}_n^{(m) \ast},
	\quad 
	m \vargeq n.
	\]
	Formally, let $\big(\Xi, \mathcal{A}, \mathbb{P}' \big)$ be a probability space assumed rich enough to carry resampling data and such replicates 
	\[
	\hat\vartheta_n^{(1) \ast}(\omega, \xi), \ldots, \hat\vartheta_n^{(m) \ast}(\omega, \xi).
	\]
	Although they may not be unconditionally independent, e.g.\ when obtained from bootstrap resampling, they are conditionally i.i.d.\ given the observed data \citep{efron1994introduction}.
	Following the standard convention in the literature,, we suppress $\omega$ and $\xi$
 with the understanding that all conditional probabilities integrate over any internal randomization used in constructing $\hat{\mathcal{I}}_{n,Q}$;
	 see also \citep[Chapter~23]{vandervaart1998asympotic}.
	Since the conditional distribution depends only on a sample of size $n$, note that large $m$ may produce many resamples that are highly redundant, and hence cannot yield additional information about the distribution of $\hat{\vartheta}_n$. Gains we will see thanks to OT can't be achieved by increasing $m$ when constructing confidence intervals only based on bootstrap.
	
	
	\paragraph{Assumptions}
	Throughout this section, we maintain the following assumptions.
	
	\begin{assumption}[Atomic Source Distribution]\label{ass:A1}
		The source $Q \in \mathcal{P}_2(\Real)$ is supported on $M$ atoms,
		\[  Q := Q_M = \sum_{i = 1}^{M} q_i \delta_{x_i} , \quad q_i>0, \quad \text{and } \sum_{i = 1}^{M} q_i = 1 \]
		with $\{Q_M\}_M$ converging weakly to an absolutely continuous probability measure in $\mathcal{P}_2(\Real)$.
	\end{assumption}
	
	\begin{assumption}\label{ass:A2}
		The admissible set $\C = \{x_{i_1}, \ldots, x_{i_k}\} \subset \spt{Q}$ is such that $Q(\C) \vargeq 1 - \alpha$.
	\end{assumption}
	
	\begin{assumption}[Weak convergence and regularity]\label{ass:A3}
		Both sequences 
		$\{P_n\}_{n\in\mathbb{N}}$, $\{\hat P_n\}_{n\in\mathbb{N}} \subset \mathcal{P}_2(\mathbb{R})$ converge weakly to the same limiting distribution $P_\infty \in \mathcal{P}_2(\mathbb{R})$.
		Moreover, $\operatorname{supp}(P_\infty)$ is a closed interval, and
		for any metric $\rho$ that metrizes weak convergence
		\[
		\rho(\hat P_n, P_\infty) \xrightarrow{\mathbb{P}} 0,
		\qquad
		\rho(P_n, P_\infty) \to 0 \ .
		\]
		By the triangle inequality, it follows that 
		$
		\rho(\hat{P}_n, P_n) \xrightarrow{\mathbb{P}} 0  .
		$
	\end{assumption}
	
	\begin{assumption}[Size and Conditional Independence]\label{ass:A4}
		For every $m \vargeq 1$, let $\hat{\vartheta}_n^{(1) \ast}, \ldots, \hat{\vartheta}_n^{(m) \ast}$ denote replicate estimators from resampling.
		We assume that conditional on the data they are i.i.d., i.e.\ for all Borel sets $A_1, \ldots, A_m \in \mathcal{B}(\Real)$
		\[
		\pr{
			\hat{\vartheta}_n^{(1) \ast} \in A_1, \ldots, \hat{\vartheta}_n^{(m) \ast} \in A_m | X_1(\omega), \ldots, X_n(\omega)
		} =
		\prod_{i = 1}^{m} \pr{\hat{\vartheta}_n^{(i) \ast} \in A_i | X_1(\omega), \ldots, X_n(\omega)}.
		\]
		Moreover, the number of replicates $m(n) = m$ is allowed to depend on $n$ and satisfies $m \to \infty$ as $n \to \infty$, possibly at a faster rate.
	\end{assumption}
	
	\begin{remark}
		In practice, the conditional distribution of the replicates 
		need not to coincide exactly with the true sampling distribution $P_n$. 
		Resampling procedures generate realizations from an approximate conditional distribution $P_n^{\ast}$, 
		and under standard regularity conditions, for a broad class of smooth statistics 
		$P_n^{\ast}$ is consistent for $P_n$ 
		\cite[Ch.~23]{vandervaart1998asympotic}. 
		Common bootstrap schemes
		satisfy these conditions. 
		Combined with Assumption~\ref{ass:A4}, this ensures that Assumption~\ref{ass:A3} is practically justified in appropriate resampling settings. 
		A more refined analysis could explicitly account for the deviation between $P_n^{\ast}$ and $P_n$, as well as the rate of convergence of $P_n^{\ast}$ to $P_n$.
	\end{remark}
	
	\begin{remark}
		Assumptions~\ref{ass:A1}--\ref{ass:A4} are sufficient for the asymptotic and finite-sample results of this paper.
		They do not require smoothness of the  OT potentials nor uniqueness of the  OT plan.
	%
		Although many of our results are formulated in terms of the convex potentials $\varphi_n$ and $\hat{\varphi}_n$, to practice one requires only knowledge of the  OT plan; see Section~\ref{sec:montecarlo}.
	\end{remark}

	\begin{remark}
		One may also want to ensure measurable selection of optimal couplings.
		In practice, however, the confidence interval is computed from a specific coupling, and its distribution or expectation is unaffected. Standard measurable selection theorems cover the general case \citep[Corollary 5.22]{villani2009old}. 
	\end{remark}
	
	\subsection{Asymptotic Properties} \label{seubsec:GAP}
	
	Now we can study the asymptotic behavior of the confidence intervals defined in \eqref{eq:3-confidence_interval_sample},
	establishing asymptotic equivalence in the Hausdorff distance which ensures validity of our method for large samples.
	For closed intervals $A = [a,b]$ and $B = [c,d]$ with $a < b$, $c < d$ the \emph{Hausdorff distance} is $	d_H(A, B) = \max \big\{ \lvert a - c \rvert, \, \lvert b - d \rvert \big\} $. We can state:
	
	
	\begin{theorem}[Asymptotic Equivalence]\label{th:3-hausdorff_consistency}
		Let $Q$ be a given atomic probability measure with $M < \infty$ atoms, and $\C$ an admissible finite set with $Q(\C) \vargeq 1 - \alpha$. Let $\hat{\mathcal{I}}_{n,Q}$ be an approximate confidence interval as in \eqref{eq:3-confidence_interval_sample}, and $\mathcal{I}_{n,Q}$ its population counterpart constructed from $P_n \in \mathcal{P}_2(\Real)$. Then, under Assumptions~\ref{ass:A1}--\ref{ass:A4}, and for all $\varepsilon > 0$, we have
		\begin{equation}\label{eq:3-hausdorff_consistency}
			\pr{d_H(\hat{\mathcal{I}}_{n,Q}, \mathcal{I}_{n,Q})>\varepsilon } \rightarrow 0 ,
			\quad	
			\text{as } n \rightarrow \infty .
		\end{equation}
\end{theorem}
	
	\begin{proof}
		Let $\varphi_{\infty}$ denote an optimal Kantorovich convex potential transporting $Q$ to $P_{\infty}$.
		Under Assumption~\ref{ass:A3} and by the stability of the discrete  OT problem, we have
		\[
		\varphi_n(x) \to \varphi_{\infty}(x),
		\quad
		\hat{\varphi}_n(x) \xrightarrow{\mathbb{P}} \varphi_{\infty}(x),
		\quad
		\forall x \in \C
		\]
		Since the admissible set $\C$ is finite, we get convergence in probability of the form
		\[   
		D^{-}\hat{\varphi}_n(c_a) \xrightarrow{\mathbb{P}} D^{-}\varphi_{\infty}(c_a) , 
		\quad
		D^{+}\hat{\varphi}_n(c_b) \xrightarrow{\mathbb{P}}D^{+}\varphi_{\infty}(c_b),	 
		\]
		and similarly for $\varphi_n$.
		It follows that
		\[
		D^{-}\hat{\varphi}_n(c_a) - D^{-}\varphi_n(c_a) \xrightarrow{\mathbb{P}} 0,
		\quad
		D^{+}\hat{\varphi}_n(c_b) - D^{+}\varphi_n(c_b) \xrightarrow{\mathbb{P}} 0,
		\]
		which yields the Hausdorff asymptotic equivalence of the intervals.		
	\end{proof}
	
	We obtain immediately the following corollaries.
	\begin{corollary}[Asymptotic Equivalence of Interval Length]
		Under the assumptions of Theorem~\ref{th:3-hausdorff_consistency}, the lengths
		\[
		\hat{\ell}_n(Q) \coloneqq r_n^{-1} \big\lvert D^{+} \hat{\varphi}_n(c_b) - D^{-} \hat{\varphi}_n(c_a)\big \rvert,
		\quad
		\ell_n(Q) \coloneqq r_n^{-1} \big\lvert D^{+} \varphi_n(c_b) - D^{-} \varphi_n(c_a)\big\rvert,
		\]
		are asymptotically equivalent, i.e.\ for all $\varepsilon > 0$,
		\begin{equation}
			\mathbb{P} \big( \lvert \hat{\ell}_n(Q) -\ell_n(Q)
			\rvert > \varepsilon \big) \longrightarrow 0 , \quad \text{as } n \to \infty .
		\end{equation}
	\end{corollary}
	
	\begin{proof}
		By Theorem~\ref{th:3-hausdorff_consistency}, the differences of both endpoints of $\hat{\mathcal{I}}_{n,Q}$ and $\mathcal{I}_{n,Q}$ converge in probability to zero.
		Since the length of an interval in one dimension is the difference of its endpoints, the claim follows.
	\end{proof}
	
	\begin{corollary}[Asymptotic Equivalence of Coverage Probability]
		Let $\hat{\pi}_n(Q) \coloneqq \mathbb{P}\big(\hat{\mathcal{I}}_{n,Q} \ni \theta_0\big)$ be the empirical coverage probability, and denote the population coverage by 
		\[
		\pi_n(Q) \coloneqq \mathbb{P}\big(\mathcal{I}_{n,Q} \ni \theta_0\big) \ .
		\]
		Then, under the assumptions of Theorem~\ref{th:3-hausdorff_consistency}, for all $\varepsilon > 0$,
		\begin{equation}
			\mathbb{P}\big(\lvert \hat{\pi}_n(Q) - \pi_n(Q) \rvert > \varepsilon \big) \longrightarrow 0, 
			\quad \text{as } n \to \infty.
		\end{equation}
	\end{corollary}
	
	\begin{proof}
		By Theorem~\ref{th:3-hausdorff_consistency}, $\hat{\mathcal{I}}_{n,Q}$ and $\mathcal{I}_{n,Q}$ are asymptotically equivalent. 	
		Since the intervals are one-dimensional, this implies that eventually the endpoints of $\hat{\mathcal{I}}_{n,Q}$ coincide with the endpoints of $\mathcal{I}_{n,Q}$ up to an arbitrarily small error. Therefore, the event that $\theta_0$ lies in one interval but not in the other occurs with probability tending to zero. 
		Hence, the difference between $\hat{\pi}_n(Q)$ and  $\pi_n(Q)$ converges to zero in probability.
	\end{proof}
	
	These results establish an asymptotic picture of the proposed confidence procedure. 
	Asymptotic Hausdorff equivalence guarantees that the estimated interval converges in shape and location to its population analogue. 
	The convergence of the interval length ensures that the scale of the uncertainty is accurately captured, ruling out asymptotic under- or over-coverage. 
	Finally, the coverage probability asymptotic equivalence shows that the procedure achieves the nominal level in large samples. All together, this demonstrates that the  OT construction yields confidence intervals that are statistically valid and geometrically stable, and thus very suitable for statistical inference.
	
	\subsection{Finite-Sample Bounds}\label{subsec:FSB}
	
	Our next objective is to select a source distribution $Q \in \mathcal{Q}$ in a data-driven way that yields confidence intervals with improved coverage probability. Our strategy is to derive finite-sample error bounds for the coverage, and to express those in terms of $Q$.
	
	We define the \emph{conditional coverage probability} with respect to a candidate $Q \in \mathcal{Q}$ as
	\begin{eqnarray}
		\label{eq:3-conditional_coverage_probability}
		& \hat{\pi}_n(Q \lvert X_1, \ldots, X_n) \coloneqq
		\pr{\hat{\mathcal{I}}_{n,Q} \ni \theta_0 \ \vert \ X_1, \ldots, X_n} = 
		\pr{\hat{\mathcal{J}}_{n,Q} \ni \vartheta_n \ \vert \ X_1, \ldots, X_n} ,
	  &  \\ &
	\text{ with }
	\hat{\mathcal{J}}_{n,Q} \coloneqq \Big[ D^{-}\hat{\varphi}_n(c_a), D^{+}\hat{\varphi}_n(c_b) \Big]  . & \nonumber 
	\end{eqnarray}
	The \emph{unconditional coverage probability} is then determined by
	\begin{equation}\label{eq:3-marginal_coverage_probability}
		\hat{\pi}_n(Q) \coloneqq \expect{\hat{\pi}_n(Q \lvert X_1, \ldots, X_n)} = \pr{\hat{\mathcal{I}}_{n,Q} \ni \theta_0 } = \pr{\hat{\mathcal{J}}_{n,Q} \ni \vartheta_n}  , 
	\end{equation}
	where the expectation is taken with respect to the data-generating process.
	Consider 
	\begin{equation}\label{eq:3-general_optimization}
		Q^{\star} \in \arg\inf_{Q \in \mathcal{Q}} \phi\big(\hat{\pi}_n(Q) -(1-\alpha)\big) ,
	\end{equation}
	where $\phi:\Real \to [ 0, \infty)$ is a symmetric and strictly convex function with $\phi(0) = 0$ and $\lim_{x \to \infty} \phi(x) = \infty$. Typical choices are $\phi(x) = \lvert x \rvert^p$, $p \vargeq 1$.
	Typically, however, exact evaluation of $\hat{\pi}_n(Q \lvert X_1, \ldots, X_n)$ and $\hat{\pi}_n(Q)$  in infeasible without knowledge of the true data generating process. One is therefore forced to use alternative means to solve \eqref{eq:3-general_optimization}, for instance by deriving stochastic finite-sample error bounds $\mathcal{E}_n(Q)$ such that
	\begin{equation}\label{eq:3-upper_bound_optimization}
		\pr{\hat{\pi}_n(Q) \vargeq 1 - \alpha - \mathcal{E}_n(Q)} \vargeq 1 - \delta,
		\quad
		\delta \in (0,1).
	\end{equation}
	Consequently, the optimization problem becomes  
	\begin{equation}\label{eq:3_general_optimization_upper_bound}
		Q^{\star} \in \arg\inf _{Q \in \mathcal{Q}} \expect{\phi\big(\mathcal{E}_n(Q)\big)}.
	\end{equation}
	
	\begin{remark}
		Although, in practice, the source distribution $Q \in \mathcal{Q}$ will be chosen  data-adaptively, for theoretical considerations $Q$ is treated as a deterministic hyper-parameter. This convention is standard in nonparametric statistics, where e.g.\ bandwidths or sieve dimension are selected in a data-driven way but treated as deterministic. 
		While the random formulation would theoretically more faithful, it would turn further analysis quite tedious, blurring the main ideas and concepts.
	\end{remark}
	
	For convenience, we proceed by bounding the \emph{probability of non-coverage},
	\begin{align}
		\hat{\overline{\pi}}_n(Q \mid X_1, \ldots, X_n) &\coloneqq \pr{\hat{\mathcal{I}}_{n,Q} \not\ni \theta_0 \mid X_1, \ldots, X_n} = \pr{\hat{\mathcal{J}}_{n,Q} \not\ni \vartheta_n \ \vert \ X_1, \ldots, X_n} \label{eq:06-conditional_non-coverage} \\
		\hat{\overline{\pi}}_n(Q) &\coloneqq \pr{\hat{\mathcal{I}}_{n,Q} \not\ni \theta_0} = \pr{\hat{\mathcal{J}}_{n,Q} \not\ni \vartheta_n}. \label{eq:06-non-coverage}
	\end{align}
	For $A^{c} = \Real \setminus A$ we can decompose (\ref{eq:06-non-coverage}) into
	\begin{equation}\label{eq:3-none_coverage_01}
		\pr{\hat{\mathcal{J}}_{n,Q} \not\ni \vartheta_n} =
		\hat{P}_n\big(\hat{\mathcal{J}}_{n,Q}^{c}\big) + 
		\Big(P_n\big(\hat{\mathcal{J}}_{n,Q}^{c}\big) - \hat{P}_n\big(\hat{\mathcal{J}}_{n,Q}^{c}\big)\Big) \varleq 
		\alpha + \Big(P_n\big(\hat{\mathcal{J}}_{n,Q}^{c}\big) - \hat{P}_n\big(\hat{\mathcal{J}}_{n,Q}^{c}\big)\Big)  ,
	\end{equation}
	with $\hat{P}_n\big(\hat{\mathcal{J}}_{n,Q}^{c}\big) \varleq \alpha$ following from the definition of $\C$ and properties of the OT problem.
	
	From this decomposition, we can give uniform bounds for the non-coverage probability.
	\begin{proposition}[Uniform Upper Bound]\label{pr:3-uniform_upper_bound}
		Let $\hat{P}_n$ be the empirical measure of $m$ replicates $\hat{\vartheta}_n^{(1) \ast}, \ldots, \hat{\vartheta}_n^{(m) \ast}$ with common probability measure $P_n \in \mathcal{P}_2(\Real)$.
		Assume Assumptions \ref{ass:A1}--\ref{ass:A4} hold.
		Then, for any atomic source probability measure $Q \in \mathcal{Q}$, any admissible set $\C$ satisfying $Q(\C) \vargeq 1 - \alpha$, and all $m \vargeq 1$, it holds that
		\begin{equation}
			\hat{\overline{\pi}}_n(Q) \varleq \alpha + 2 \sqrt{\frac{\ln(2/\delta)}{2m}},
		\end{equation}
		with probability of at least $1 - \delta$ for a fixed  $\delta \in (0,1)$.
	\end{proposition}
	\begin{proof}
		Fix measure $Q \in \mathcal{Q}$ and an admissible set $\C$.
		From Equation \eqref{eq:3-none_coverage_01}, we have 
		\[	\hat{\overline{\pi}}_n(Q) \varleq \alpha + \lvert P_n(\hat{\mathcal{J}}_{n,Q}^{c}) - \hat{P}_n(\hat{\mathcal{J}}_{n,Q}^{c}) \rvert . \]
		Since $\hat{\mathcal{J}}_{n,Q}^{c}$ is the disjoint union of two rays, this difference can be expressed as the sum of two deviations of the empirical c.d.f.\ from the true c.d.f.\ at the interval endpoints. Hence,
		\[	\lvert P_n(\hat{\mathcal{J}}_{n,Q}^{c}) - \hat{P}_n(\hat{\mathcal{J}}_{n,Q}^{c}) \rvert \varleq 2 \sup_{x \in \Real} \lvert \hat{F}_n(x) - F_n(x) \rvert .	\]
		Applying the Dvoretzky-Kiefer-Wolfowitz-Massart (DKW-M) inequality conditionally on the observed data and taking expectation yields, with probability of at least $1 - \delta$,
		\[	\sup_{x\in\Real} \lvert \hat{F}_n(x)-F_n(x)\rvert \varleq \sqrt{\frac{\ln(2/\delta)}{2m}}.\]
		Combining the preceding bounds gives
		\[
		\hat{\overline{\pi}}_n(Q) \varleq \alpha + 2 \sqrt{\frac{\ln(2/\delta)}{2m}},
		\]
		with probability of at least $1-\delta$.
		Since the choice of $Q$ and $\mathcal{C}$ was arbitrary, the bound holds uniformly over all $Q \in \mathcal{Q}$ and admissible sets $\mathcal{C}$.
	\end{proof}
	
	The bound holds simultaneously for all $Q \in \mathcal{P}_2(\Real)$ and is distribution-free. More precisely, since $\mathcal{I}_{n,Q}$ is always a closed interval, 
	its complement consists of at most two rays.
	Uniform control of the empirical measure deviation at the interval endpoints 
	yields a bound that does not depend on the specific form of $Q$.
	
	\begin{remark}
		If the bootstrap replicates are conditionally i.i.d.\ from an
		approximate distribution $P_n^*$ instead of $P_n$, the proof gives
		\[
		\hat{\overline{\pi}}_n(Q)
		\varleq
		\alpha
		+
		2\sqrt{\frac{\ln(2/\delta)}{2m}}
		+
		2\sup_{x\in\mathbb{R}} |F_n(x)-F_n^*(x)|.
		\]
		The last term quantifies the discrepancy between the idealized
		distribution $P_n$ and the bootstrap distribution $P_n^*$.
		Under standard regularity conditions for smooth statistics, classical bootstrap theory ensures that
		\[
		\sup_{x\in\mathbb{R}} |F_n(x)-F_n^*(x)| \to 0,
		\quad
		\text{as } n \to \infty,
		\]
		in probability, and $\sup_{x\in\mathbb{R}} |F_n(x)-F_n^*(x)| = O_{\mathbb{P}}(n^{-1/2})$; see
		\cite{vandervaart1998asympotic}. 
	\end{remark}
	
	\begin{corollary}\label{crll:expectation_noncoverage_uniform}
		Under the same setting and assumptions as Proposition~\ref{pr:3-uniform_upper_bound}, the expected non-coverage probability of the interval satisfies
		\begin{equation}
			\expect{\hat{\overline{\pi}}_n(Q)} \varleq \alpha +  \frac{1}{\sqrt{m}}.
		\end{equation}
	\end{corollary}
	This follows from classical bounds on the expectation of $\sup_{x\in\Real} \lvert \hat{F}_n(x)-F_n(x)\rvert$; see~\citep[Theorem 3.3]{devroye2001combinatorial}.
	
	Although Proposition~\ref{pr:3-uniform_upper_bound} and Corollary~\ref{crll:expectation_noncoverage_uniform} provide uniform finite-sample upper bounds on the non-coverage probability, they mainly serve as preliminary guarantees as we are interested in small samples. 
	Because the DKW-M inequality controls deviations uniformly \citep{massart1990tight} , 
	these bounds are not expressed in terms of the specific source measure $Q$ and therefore reflect a worst-case scenario rather than a distribution or design-specific guarantees. 
	Nevertheless, these results remain informative: they certify that the coverage error decays at the parametric rate $m^{-1/2}$, ensuring stable asymptotic behavior even in the worst case. 
	However, our analysis is not strictly minimax; by selecting appropriately a source measure $Q$ in the  OT construction, the bound of coverage error can be improved for specific distributions, although the worst-case constant remains unattainable.
	
	Let us refine the upper bound to make its dependence on $Q$ more explicit.  
	Recalling \eqref{eq:3-none_coverage_01} and using that $Q$ has finite support, we obtain
	\[
	P_n\big(\hat{\mathcal{J}}_{n,Q}^{c}\big) - \hat{P}_n\big(\hat{\mathcal{J}}_{n,Q}^{c}\big) \varleq
	Q\big(\mathcal{S}^c\big) - Q\big(\C^{c}\big) =
	\sum_{i \in \mathcal{S}^c} q_i - \sum_{i \in \C^{c}} q_i ,
	\]
	where the random set $\mathcal{S}^c$ is defined as
	\begin{equation}\label{eq:inverse_set_confidence_interval}
		\mathcal{S}^c \coloneqq
		\Big\{x \in \spt{Q} : \exists y \in \hat{\mathcal{J}}_{n,Q}^{c} \text{ with } y \in \partial \varphi_n(x)\Big\}.
	\end{equation}
	The first term represents undercoverage and the second one overcoverage. Define
	\begin{equation}\label{eq:3delta}
		\mathcal{E}_n(Q) \coloneqq \sum_{i = 1}^{M} Y_i,
		\quad
		Y_i \coloneqq q_i(\mathbbm{1}_{i \in \mathcal{S}^c} - \mathbbm{1}_{i \in \C^c}),
	\end{equation}
	such that the non-coverage probability satisfies the almost-sure bound
	\[
	\hat{\overline{\pi}}_n(Q) \varleq \alpha + \mathcal{E}_n(Q)
	\quad
	\text{a.s.}
	\]
	This gives explicit upper bound for the non-coverage probability in terms of $Q$.
	\begin{theorem}[Upper Bound for Non-Coverage Probability]\label{th:3-concentration_noncoverage}
		Let $\hat{P}_n \in \mathcal{P}_2(\Real)$ be a probability measure for $m$ replicates $\hat\vartheta_n^{(1) \ast}, \ldots, \hat\vartheta_n^{(m) \ast}$ with common distribution $P_n \in \mathcal{P}_2(\Real)$.
		Let $\mathcal{S}^{c}$ be as in \eqref{eq:inverse_set_confidence_interval}, and $\mathcal{E}_n(Q)$ be as in \eqref{eq:3delta}.
		Assume Assumptions~\ref{ass:A1}--\ref{ass:A4} hold, then for any source probability measure $Q \in \mathcal{Q}$, any admissible set $\C$ satisfying $Q(\C) \vargeq 1 - \alpha$, and all $m \vargeq 1$, the following bound holds
		\begin{equation}
			\hat{\overline{\pi}}_n(Q) \varleq 
			\alpha + \expect{\mathcal{E}_n(Q)} + \sqrt{2 \sum_{i = 1}^{M} q_i^2 \ln(1/\delta)},
		\end{equation}
		with probability of at least $1 - \delta$ for any fixed $\delta \in (0,1)$.
	\end{theorem}
	
	\begin{proof}
		Fix $Q \in \mathcal{Q}$ and an admissible set $\C$.
		Let $\mathcal{F}_i \coloneqq \sigma(Y_1,\ldots,Y_i)$ be the filtration generated by $Y_1,\ldots,Y_i$, and define the martingale
		\[
		Z_i \coloneqq \expect{\mathcal{E}_n(Q) \mid \mathcal{F}_i}, 
		\qquad i = 0,1,\ldots,M,
		\]
		so that $Z_0 = \expect{\mathcal{E}_n(Q)}$ and $Z_M = \mathcal{E}_n(Q)$. Since $\lvert Y_i \rvert \varleq q_i$ almost surely, it follows for the martingale increments $\Delta_i \coloneqq Z_i - Z_{i-1}$ that
		\[
		\lvert \Delta_i \rvert \varleq q_i 
		\qquad \text{a.s.\ for each } i=1,\ldots,M.
		\]
		Thus, $\{Z_i\}_{i=0}^{M}$ is a martingale with bounded differences. 
		By the inequality of \cite{azuma1967weighted} and \cite{hoeffding1963probability}, we know that for any $t>0$,
		\[
		\pr{ Z_M - Z_0 \vargeq t }  =
		\pr{ \mathcal{E}_n(Q) - \expect{\mathcal{E}_n(Q)} \vargeq t }
		\varleq
		\exp\!\Big(-\frac{t^2}{2\sum_{i=1}^{M} q_i^2}\Big).
		\]
		Setting $t = \sqrt{2 \sum_{i=1}^{M} q_i^2 \, \ln(1/\delta)}$ gives, with probability of at least $1-\delta$,
		\[
		\mathcal{E}_n(Q) 
		\varleq 
		\expect{\mathcal{E}_n(Q)} 
		+  \sqrt{2 \sum_{i=1}^{M} q_i^2 \, \ln(1/\delta)}.
		\]
		Since $Q$ and $\C$ were arbitrary, the stated bound holds uniformly for all $Q \in \mathcal{Q}$ and all admissible sets $\C$. This completes the proof.
	\end{proof}
	
	\begin{remark}
		In the special case where $Q$ is uniform on $M$ atoms, we have $q_i = 1/M$ for all $i$, and therefore $\sum_{i=1}^{M} q_i^2 = 1/M$. 
		In this case, the bound reduces to
		\[
		\hat{\overline{\pi}}_n(Q) \varleq 
		\alpha + \mathbb{E}[\mathcal{E}_n(Q)] 
		+ \sqrt{\frac{2 \ln(1/\delta)}{M}}.
		\]
	\end{remark}
	
	The bound in Theorem~\ref{th:3-concentration_noncoverage} provides a finite-sample control of $\hat{\pi}_n(Q)$ that depends explicitly on the atomic weights of $Q$. 
	Given the generality of this setting, we do not expect substantially tighter bounds without imposing additional structural assumptions on $P_n$ or $Q$. 
	In this sense, the result yields a worst-case 
	together with a control of undercoverage and overcoverage.
	Moreover, the expectation of $\mathcal{E}_n(Q)$ admits a closed-form expression in terms of $Q$. 	By linearity of expectation,
	\begin{equation}
		\expect{\mathcal{E}_n(Q)}
		= \sum_{i=1}^{M} q_i\,\pr{x_i \in \mathcal{S}^c}
		- \sum_{i \in \C^c} q_i .
	\end{equation}
	To improve the bounds, some technical controls are required.
	The next Lemma quantifies how discrepancies between the inverse c.d.f.\ and its empirical counterpart control deviations of the extreme subdifferential points.
	\begin{lemma}\label{lem:empirical_extreme_deviation}
		Let $\hat{\mu}_n$ be the empirical measure from an i.i.d.\ sample $U_1, \ldots, U_n$ with common probability measure $\mu \in \mathcal{P}_2(\Real)$.
		Let $Q \in \mathcal{Q}$ be arbitrary, and let $\gamma \in \Gamma(Q,\mu)$ and $\hat{\gamma}_n \in \Gamma(Q,\hat{\mu}_n)$ be optimal couplings with optimal convex potentials $\varphi$ and $\hat{\varphi}_n$, respectively.
		Define the extreme points of the subdifferentials by
		\[
		u_i^{-} \coloneqq D^{-} \varphi(x_i),
		\quad 
		u_i^{+} \coloneqq D^{+}\varphi(x_i),
		\quad
		\hat{u}_i^{-} \coloneqq D^{-} \hat{\varphi}_n(x_i),
		\quad 
		\hat{u}_i^{+} \coloneqq D^{+}\hat{\varphi}_n(x_i).
		\]
		Let $V_i \coloneqq \sum_{j = 1}^{i} q_j$, with $V_0 = 0$.
		Then, for all $i = 1, \ldots, M$, and for every realization of the sample
		\begin{equation}
			\lvert u_i^{\pm} - \hat{u}_i^{\pm} \rvert \varleq
			\sup_{0 \varleq u \varleq 1} \lvert F^{-1}(u) - \hat{F}_n^{-1}(u)\rvert \ ,
		\end{equation}
		where $F^{-1}$, $\hat{F}_n^{-1}$ denote the generalized inverses of the c.d.f.\ and the empirical c.d.f.\ of $\mu$.
	\end{lemma}
	
	The following lemma provides a quantitative upper bound on the probability that an atom $x_i$ is in $\mathcal{S}^{c}$, a key step in controlling the non-coverage probability.
	\begin{lemma}[Upper Bound on the Probability]\label{lem:3-upper_bound_probability}
		Let $\hat{P}_n$ be the empirical probability measure from $m$ replicates $\hat\vartheta_n^{(1) \ast}, \ldots, \hat\vartheta_n^{(m) \ast}$ with common probability measure $P_n \in \mathcal{P}_2(\Real)$.
		Let $Q \in \mathcal{Q}$ be arbitrary, and $\C$ any admissible set with $Q(\C) \vargeq 1 - \alpha$ and $c_a \coloneqq \min \C$, $c_b \coloneqq \max \C$.
		Let $\gamma_n \in \Gamma(Q, P_n)$ and $\hat{\gamma}_n \in \Gamma(Q, \hat{P}_n)$ be optimal couplings with optimal convex potentials $\varphi_n$, $\hat{\varphi}_n$  respectively.
		Define $u_i^{-}$, $u_i^{+}$, 
		as in Lemma~\ref{lem:empirical_extreme_deviation}, further 
		\[
		u_a \coloneqq D^{-}\varphi_n(c_a), 
		\quad
		u_b \coloneqq D^{+}\varphi_n(c_b), 
		\text{ and }  g_i \coloneqq	\min \{ u_i^{-} - u_a, u_b - u_i^{+}\}   \text{ (deterministic gap)}.
		\]
		Then, for $\mathcal{S}^{c}$ as in \eqref{eq:inverse_set_confidence_interval}, and for any $i = 1, \ldots, M$, it holds almost surely that 
		\begin{equation}
			\pr{x_i \in \mathcal{S}^c} \varleq
			\pr{\supquant{u \in [0, 1]} \vargeq \frac{g_i}{2}}\mathbbm{1}_{x_i \in \C} + \mathbbm{1}_{x_i \in \C^c} \ ,
		\end{equation}
		where $F_n^{-1}$ and $\hat{F}_n^{-1}$ denote the generalized quantile function of $P_n$ and $\hat{P}_n$, respectively.
	\end{lemma}
	Here, the factor $1/2$ is a geometric constant coming from the \emph{half-gap} argument, c.f.\ the proof(s) in the Appendix. 
	%
	The following theorem provides an explicit finite-sample upper bound on the non-coverage probability that depends on the source measure $Q$.
	\begin{theorem}[Upper Bound on the Non-Coverage Probability II]\label{th:3-upper_bound_noncoverage_II}
		Let $\hat{P}_n$ be the empirical measure of $m$ replicates $\hat{\vartheta}_n^{(1) \ast}, \ldots, \hat{\vartheta}_n^{(m) \ast}$ with common distribution $P_n \in \mathcal{P}_2(\Real)$, fix $\delta \in (0,1)$.
		Let $\mathcal{S}^{c}$ be as in Equation~\eqref{eq:inverse_set_confidence_interval}.
		Assume Assumptions \ref{ass:A1}--\ref{ass:A4} hold, and that  $P_n$ 
		\begin{enumerate}
			\item either is supported on a closed interval of length $L < \infty$ with density $f(x) \vargeq f_{\min} > 0$ for all $x \in \spt{P_n}$,
			\item or has a finite Laplace transform in a neighborhood of 0, i.e.\ $\exists$  $t_n > 0$ such that
			\[
			\mathcal{T}_{\lambda}(P_n) \coloneqq \int_{\Real} e^{\lambda x} dP_n(x) < \infty,
			\quad
			\forall \lambda \in [-t_n, t_n].
			\]
			Moreover, $\exists \varepsilon_n \in (0, 1/2)$ with $\varepsilon_n \to 0$ as $n \to \infty$, such that $P_n$ has a density $f(x) \vargeq f_{\min} > 0$ for all $x \in [F_n^{-1}(\varepsilon_n), F_n^{-1}(1 - \varepsilon_n)]$.
		\end{enumerate}
		Then, for any source probability measure $Q \in \mathcal{Q}$, any admissible set $\C$ satisfying $Q(\C) \vargeq 1 - \alpha$ and all $m \vargeq 1$, the following bound holds
		\begin{equation}
			\hat{\overline{\pi}}_n(Q) \varleq 
			\alpha + 
			\sum_{i \in \C} 
			q_i\Big(2\exp\big(-\frac{m f_{\min}^2 g_i^2 }{2}\big) + A(g_i) \Big)  +
			\sqrt{2 \ln(1/\delta) \sum_{i = 1}^{M} q_i^2},
		\end{equation}
		with probability of at least $1 - \delta$, and
		\[
		A(g_i) =
		\begin{cases}
			0, \quad &\text{under 1.} \\
			2m \displaystyle\inf_{\lambda \in [0, t_n]} \Big\{  \big(\mathcal{T}_{\lambda}(P_n) + \mathcal{T}_{-\lambda}(P_n)\big)\exp(- \lambda \frac{g_i}{2}) \Big\}, \quad &\text{under 2.}
		\end{cases}
		\]
		and where $g_i < 2L$ under 1.\ (otherwise it is trivially 0).
	\end{theorem}
	
\begin{proof}
	The proof is given for case (2.), as for (1.) it follows directly from the compact support assumption.
	All probabilities are taken conditionally on the data.
	Fix a source probability measure $Q \in \mathcal{Q}$, and an admissible set $\C$.
	It obviously suffices to address the $x_i$ in $\C$.
	From Lemma \ref{lem:3-upper_bound_probability}, we have for all $i = 1, \ldots, M$
	\[	\pr{x_i \in \mathcal{S}^c} \varleq 
	\pr{\supquant{u \in [0, 1]} \vargeq \frac{g_i}{2}}\mathbbm{1}_{x_i \in \C} + \mathbbm{1}_{x_i \in \C^c} .
	\]
	Decompose the probability of interest into
	\[  
	\pr{\supquant{u \in [\varepsilon_n, 1 - \varepsilon_n]} \vargeq \frac{g_i}{2}} +
	\pr{\supquant{u \in [0, \varepsilon_n) \cup (1 - \varepsilon_n, 1]} \vargeq \frac{g_i}{2}} 
	\]
	and consider first the central part.
	Since $f(x) \vargeq f_{\min} > 0$, $F_n$ satisfies
	\[
	\big\lvert F_n(x) - F_n(x') \big\rvert \vargeq f_{\min} \lvert x - x' \rvert ,
	\qquad
	\forall x, x' \in \Big[F_n^{-1}(\varepsilon_n), F_n^{-1}(1 - \varepsilon_n)\Big] .
	\]
	Hence, its inverse is locally $1/f_{\min}$-Lipschitz, and
	\[
	\sup_{u \in [\varepsilon_n, 1 - \varepsilon_n]} \big\lvert F_n^{-1}(u) - \hat{F}_n^{-1}(u) \big\rvert \varleq
	\frac{1}{f_{\min}} \sup_{x \in \Real} \big\lvert F_n(x) - \hat{F}_n(x) \big\rvert .
	\]
	By the DKW-M inequality,
	\[
	\pr{ \!\!\! \supquant{u \in [\varepsilon_n, 1 - \varepsilon_n]} \vargeq \frac{g_i}{2}} \varleq
	\pr{\sup_{x \in \Real} \lvert F_n(x) - \hat{F}_n(x) \rvert \vargeq \frac{f_{\min}g_i}{2}} \varleq 	2\exp\Big(\frac{m f_{\min}^2 g_i^2}{-2}\Big) 
	\]
	if $g_i / 2 < L$ (otherwise it is 0).
	
	It remains to control the tail regions $[0, \varepsilon_n) \cup (1 - \varepsilon_n, 1]$.
	If
	\[
	\supquant{u \in [0, \varepsilon_n)}\vargeq \frac{g_i}{2},
	\]
	then there exists at least one $j \in \{1, \ldots, m\}$ such that $\lvert \hat{\vartheta}_n^{(j) \ast} \rvert \vargeq g_i/2$.
	Applying an union bound we have
	\[
	\pr{\supquant{u \in [0, \varepsilon_n)} \vargeq \frac{g_i}{2}}\varleq 
	\pr{\exists j \in \{1, \ldots, m\} : \lvert \hat\vartheta_n^{(j)} \rvert \vargeq \frac{g_i}{2}} \varleq
	m \pr{\lvert \vartheta_n \rvert \vargeq \frac{g_i}{2}},
	\]
	and the same inequality holds for the right tail.
	Therefore,
	\[
	\pr{\supquant{u \in [0, \varepsilon_n) \cup (1 - \varepsilon_n, 1]} \vargeq \frac{g_i}{2}} \varleq
	2m  \pr{\lvert \vartheta_n \rvert \vargeq\frac{g_i}{2}}
	\]
	From the finite Laplace transform assumption and Markov's inequality, and for all $t > 0$ and $\lambda \in [0, t_n]$
	\[
	\pr{\vartheta_n \vargeq t} \varleq \mathcal{T}_{\lambda}(P_n)\exp(-\lambda t),
	\quad
	\pr{\vartheta_n \varleq t} \varleq \mathcal{T}_{-\lambda}(P_n)\exp(-\lambda t).
	\]
	Hence,
	\[
	\pr{\lvert \vartheta_n \rvert \vargeq t} \varleq 
	\big(\mathcal{T}_{\lambda}(P_n) + \mathcal{T}_{-\lambda}(P_n)\big)\exp(-\lambda t).
	\]
	Combining these results with the central part, optimizing over $\lambda \in [-t_n, t_n]$, and taking expectation with respect to the data yields the claimed bound.
	Since the choice of $Q$, and $\C$ was arbitrary, the result holds for all $Q \in \mathcal{Q}$, and all admissible sets $\C$.
\end{proof}

	\begin{remark}
		The requirement that $\mathcal{T}_{\lambda}(P_n) < \infty$ for all $\lambda$ in a neighborhood of $0$ is standard, as it ensures the existence of exponential moments. It is satisfied by many common distributions (Gaussian, exponential, Gamma, etc.) and is weaker than requiring $P_n$ to be sub-Gaussian. This assumption allows one to derive exponential bounds via a Markov–Chernoff argument while remaining mild for most statistical models of interest.
		To see how different tail conditions on $P_n$ affect the term $A(g_i)$ in Theorem~\ref{th:3-upper_bound_noncoverage_II}, consider
		
		\begin{itemize}
			\item a sub-Gaussian $P_n$ with parameter $v > 0$, so that $\pr{\lvert \vartheta_n \rvert > t} \varleq 2 \exp(-t^2/2v)$, then
			\[
			A(g_i) \lesssim m \exp(-c g_i^2) , \quad \text{ for some positive constant $c$.}
			\]
			
			\item when there exists $0 < \beta < \infty$ such that
			\[
			\mathcal{T}_{\lambda, \beta}(P_n) \coloneqq \int_{\Real} e^{\lambda x^\beta} \mathrm{d}P_n(x) < \infty,
			\quad
			\forall \lambda \text{ in a neighborhood of 0},
			\]
			then $\pr{\lvert \vartheta_n \rvert > t} \lesssim \exp(-\lambda t^\beta)$, and 
			$ 
			A(g_i) \lesssim m \exp(-c g_i^\beta) , \text{ for a positive constant $c$.}
			$ 
			
			\item a $P_n$ with only finite $p$ moments; Markov's inequality implies polynomial tails, giving
			\[
			A(g_i) \lesssim m g_i^{-p}.
			\]
		\end{itemize}
	\end{remark}
	
	
	\begin{remark}\label{rmrk:uniform_gaussian}
		Consider $P_n = \mathcal{N}(0,\sigma_n^2)$, and let $Q$ be uniform with $q_i = 1/M$. Then the Laplace transform can be computed explicitly:
		\[
		A(g_i) = 2 m \inf_{\lambda \in [0,t_n]} \big\{ (\mathcal{T}_\lambda(P_n) + \mathcal{T}_{-\lambda}(P_n)) \exp(- \lambda g_i/2) \big\}
		= 4 m \inf_{\lambda \in [0,t_n]} \exp\Big(\frac{\sigma_n^2 \lambda^2 - \lambda g_i}{2}\Big),
		\]
		which is minimized at $\lambda^\star = g_i / (2\sigma_n^2)$ when $\lambda^\star \varleq t_n$, yielding
		\[
		A(g_i) = 4 m \exp\Big(- \frac{g_i^2}{8\sigma_n^2}\Big).
		\]
		Combined with the exponential term $2 \exp(- m f_{\min}^2 g_i^2 / 2)$ and the last term $\sqrt{2 \ln(1/\delta)/M}$, all contributions beyond $\alpha$ vanish asymptotically provided
		\[
		\sigma_n^2 \ll \min_i g_i^2, \quad M \to \infty, \quad m \to \infty.
		\]
		Hence, the bound satisfies $\hat{\overline{\pi}}_n(Q) = \alpha + o(1)$. Uniform weights $q_i = 1/M$ suffice for the asymptotic vanishing of the bound in the Gaussian case.
	\end{remark}

	\begin{corollary}\label{crll:expectation_noncoverage_II}
		Under the same setting and assumptions as Theorem~\ref{th:3-upper_bound_noncoverage_II}, the expected non-coverage probability satisfies
		\begin{equation}
			\expect{\hat{\overline{\pi}}_n(Q)} \varleq
			\alpha + 
			\sum_{i \in \C} 
			q_i\Big(2\exp\big(-\frac{m f_{\min}^2 g_i^2 }{2}\big) + A(g_i) \Big) + \sqrt{\frac{\pi}{2}\sum_{i = 1}^{M}q_i^2},
		\end{equation}
		where
		\[
		A(g_i) =
		\begin{cases}
			0, \quad &\text{under 1.} \\
			2m \displaystyle\inf_{\lambda \in [0, t_n]} \Big(\mathcal{T}_{\lambda}(P_n) + \mathcal{T}_{-\lambda}(P_n)\Big)\exp(- \lambda \frac{g_i}{2}), \quad &\text{under 2.}
		\end{cases}
		\]
	\end{corollary}
	
	Theorem~\ref{th:3-upper_bound_noncoverage_II} gives a finite-sample upper bound on the non-coverage probability, explicitly showing how it depends on measure $Q$ and the tail properties of $P_n$. 
	Corollary~\ref{crll:expectation_noncoverage_II} translates this into a bound on the expected non-coverage probability, highlighting how the weights $q_i$ and gaps $g_i$ control the average coverage performance of the interval.

	\section{Length and Data-Driven Choices} \label{sec:miscel}
	
	In the previous section, we provided bounds showing how the choice of the source distribution $Q$ affects the coverage probability.
	To complement these results, we now provide finite-sample bounds on the length of the confidence intervals, highlighting the trade-off between interval width and coverage.
	From these we derive practical, data-driven strategies for selecting all hyper parameters needed for the interval construction.
	
	\subsection{Length of the Confidence Interval} \label{subsec:Length}
	For a fixed nominal coverage $( 1-\alpha )$, shorter confidence intervals are generally preferable.
	It is therefore natural to investigate how the length of OT-based confidence intervals depend on the choice of $Q$.
	Because our framework allows arbitrary discrete source distributions with finite support, classical regularity assumptions are not relevant.
	Consequently, we focus on uniform, distribution-free bounds that hold for all $Q \in \mathcal{Q}$.
	
	Throughout this subsection, we assume that
	\[
	\spt{P_n} \subseteq [a,b],
	\quad
	L = b - a,
	\]
	so that the statistic of interest satisfies $\lvert \vartheta_n \rvert \varleq K \coloneqq \max\big(\lvert a \rvert, \lvert b \rvert \big)$ almost surely.
	Let
	\begin{equation}\label{eq:3-length}
		\hat{\ell}_n(Q) \coloneqq r_n^{-1} \big\lvert D^{+} \hat{\varphi}_n(c_b)-D^{-} \hat{\varphi}_n(c_a) \big\rvert
 	\end{equation}
	denote the empirical confidence interval length, and $\ell_n(Q)$ its population analogue.
	
	\begin{proposition}(Two-Sided Bound for the Length)\label{th:3-concentration_length}
		Let $\hat{P}_n$ be the empirical measure of $m$ replicates $\hat{\vartheta}_n^{(1) \ast}, \ldots, \hat{\vartheta}_n^{(m) \ast}$ with common distribution $P_n \in \mathcal{P}_2(\Real)$ and fix $\delta \in (0,1)$.
		Assume $\lvert \vartheta_n \rvert \varleq K$ almost surely, and that
		Assumptions~\ref{ass:A1}--\ref{ass:A4} hold. Then for any $Q \in \mathcal{Q}$, any admissible set $\C$ with $Q(\C) \vargeq 1 - \alpha$, and for any $m \vargeq 1$, one has the upper bound
		\[	
		\big\lvert \hat{\ell}_n(Q) - \ell_n(Q) \big\rvert \varleq 
		2r_n^{-1}\Big(\expect{\sup_{0 \varleq u \varleq 1} \big\lvert F_n^{-1}(u) - \hat{F}_n^{-1}(u) \big\rvert} + K\sqrt{\frac{\ln(2/\delta)}{2}}\Big) 
		\]
		with probability of at least $1 - \delta$.
	\end{proposition}
	\begin{proof}
		We first condition on the observed data $X_1, \ldots, X_n$. Fix an arbitrary source measure $Q \in \mathcal{Q}$ with an admissible set $\C$.
		By assumption, $\lvert \vartheta_n \rvert \varleq K$ almost surely, hence the length satisfies $\lvert \hat{\ell}_n(Q) \rvert \varleq r_n^{-1} K$.  
		Applying Hoeffding's inequality conditionally on the data, and taking expectation gives
		\[
		\big\lvert \hat{\ell}_n(Q) - \expect{\hat{\ell}_n(Q)} \big\rvert \varleq 2 r_n^{-1} K \sqrt{\frac{\ln(2/\delta)}{2}},
		\]
		with probability of at least $1 - \delta$.	
		Next, using the triangle inequality, we have
		\[
		\hat{\ell}_n(Q) \varleq r_n^{-1} \Big( 
		\big\lvert D^+ \hat{\varphi}_n(c_b) - D^+ \varphi_n(c_b) \big\rvert + 
		\big\lvert D^- \hat{\varphi}_n(c_a) - D^- \varphi_n(c_a) \big\rvert \Big) +
		r_n^{-1}  \big\lvert  D^+ \varphi_n(c_b) - D^- \varphi_n(c_a) \big\rvert 
		\]	
   denoting the last term by $\ell_n(Q)$
		By Lemma~\ref{lem:empirical_extreme_deviation}, the deviations of the empirical potentials from the population potentials are bounded by
		\[
		\big\lvert D^+ \hat{\varphi}_n(c_b) - D^+ \varphi_n(c_b) \big\rvert, \ 
		\big\lvert D^- \hat{\varphi}_n(c_a) - D^- \varphi_n(c_a) \big\rvert \varleq 
		\sup_{0 \varleq u \varleq 1} \big\lvert F_n^{-1}(u) - \hat{F}_n^{-1}(u) \big\rvert.
		\]
		Combining the above inequalities yields
		\[
		\hat{\ell}_n(Q) \varleq \ell_n(Q) + 2 r_n^{-1} \sup_{0 \varleq u \varleq 1} \big\lvert F_n^{-1}(u) - \hat{F}_n^{-1}(u) \big\rvert.
		\]
		Taking expectation and adding the Hoeffding bound gives the stated result:
		\[
		\big\lvert \hat{\ell}_n(Q) - \ell_n(Q) \big\rvert \varleq 2 r_n^{-1} \Big( \expect{\sup_{0 \varleq u \varleq 1} \big\lvert F_n^{-1}(u) - \hat{F}_n^{-1}(u) \big\rvert} + K \sqrt{\frac{\ln(2/\delta)}{2}} \Big),
		\]
		with probability of at least $1 - \delta$.
		Since $Q$ and $\C$ were arbitrary, the result holds for all $Q \in \mathcal{Q}$ and all admissible sets $\C$.
	\end{proof}

	\begin{remark}
		The boundedness condition $\lvert \vartheta_n \rvert \varleq K$ is primarily a technical assumption. In practice, one can safely truncate $\vartheta_n$ at a sufficiently large $K$, which has a negligible impact on the 
		behavior of the confidence interval length.
	\end{remark}
	
	It remains to control the expectation of 
	\[
	\sup_{0 \varleq u \varleq 1} \big\lvert F_n^{-1}(u) - \hat{F}_n^{-1}(u) \big\rvert,
	\]
	for deriving finite-sample bounds on the length of the confidence interval.
	\begin{lemma}\label{lem:expectation_sup-quantile}
		Let $\hat{\mu}_n$ be the empirical measure from an i.i.d.\ sample $U_1, \ldots, U_n$ with common distribution $\mu \in \mathcal{P}(\Real)$.
		Let $F(t)$ and $\hat{F}_n(t)$ denote the c.d.f. of $\mu$ and $\hat{\mu}_n$, respectively.
		Assume that $\mu$ has compact support of diameter $0 < L < \infty$, with a density $f$ satisfying
		\[
		f(x) \vargeq f_{\min} > 0, \quad \forall x \in \spt{\mu}.
		\]
		
		Then, the expected maximal deviation between the quantile functions satisfies
		\begin{equation}
			\expect{\sup_{0 \varleq u \varleq 1} \big\lvert F^{-1}(u) - \hat{F}_n^{-1}(u) \big\rvert} \varleq \frac{C}{\sqrt{n}},
		\end{equation}
		where the constant $C > 0$ is explicitly given by
		\begin{equation}  \label{def-C}
		C = \frac{\sqrt{\ln 2} + \sqrt{\pi}}{\sqrt{2} f_{\min}} \ .
		\end{equation}
	\end{lemma}
	For a proof see the Appendix or 
	\citep
	{bobkov2019onedimensional}.
	
	\begin{corollary}[Two-Sided Bound for the Length II]\label{th:3-2_bound_length_II}
		Let $\hat{P}_n$ be the empirical measure of $m$ replicates $\hat{\vartheta}_n^{(1) \ast}, \ldots, \hat{\vartheta}_n^{(m) \ast}$ with common distribution $P_n \in \mathcal{P}_2(\Real)$ and fix $\delta \in (0,1)$.
		Assume $\lvert \vartheta_n \rvert \varleq K$ a.s.\ with a density $f$ satisfying $f(x) \vargeq f_{\min} > 0$ for all $x \in \spt{P_n}$, and that
		Assumptions \ref{ass:A1}--\ref{ass:A4} hold.
		Then, for any source probability measure $Q \in \mathcal{Q}$, any admissible set $\C$ satisfying $Q(\C) \vargeq 1 - \alpha$ and all $m \vargeq 1$, the following bound holds
		\begin{equation}
			\big\lvert \hat{\ell}_n(Q) - \ell_n(Q) \big\rvert \varleq 
			2r_n^{-1}\Big(\frac{C}{\sqrt{m}} + K\sqrt{\frac{\ln(2/\delta)}{2}}\Big),
		\end{equation}
		with probability at least $1 - \delta$, and constant $C$ as in (\ref{def-C}).
	\end{corollary}
	
	\begin{corollary} \label{crl:expectation_length_II}
		Under the same setting and assumptions as in Corollary~\ref{th:3-2_bound_length_II} and Lemma~\ref{lem:expectation_sup-quantile}, 
		\begin{equation}
			\expect{\big\lvert \hat{\ell}_n(Q) - \ell_n(Q) \big\rvert} \varleq
			2r_n^{-1}\Big(\frac{C}{\sqrt{m}} + K\sqrt{\frac{\pi}{2}}\Big)  \ . 
		\end{equation}
	\end{corollary}

	\subsection{Data-Driven Choices for the Source Distribution} \label{subsec:Choices}
	
	Now we develop a strategy for selecting a suitable source distribution $Q$. 
	Although Theorem~\ref{th:3-upper_bound_noncoverage_II} provides insight into the non-coverage probability, it is not directly useful for optimization, since the gaps $g_i$ must be estimated, and the admissible class $\mathcal{Q}$ of atomic source measures is too large to search over. 
	We therefore restrict our selection to the subclass $\mathcal{Q}_{\beta} \subset \mathcal{Q}$ of discretized Beta distributions with $1 \varleq M < \infty$ equally spaced support points on $[0,1]$ and weights determined by shape parameters $(\eta, \zeta)$.
	This choice is motivated by the geometry of $P_n$ and Remark~\ref{rmrk:uniform_gaussian}.
	Let $\Phi$ denote the c.d.f.\ of a fixed symmetric reference distribution (e.g., Gaussian or Student's $t$) chosen to reflect the central tendency and scale of $\vartheta_n$. 
	If $\vartheta_n$ is symmetric, then the push-forward measure $\Phi_\# P_n$ is roughly uniform on $[0,1]$, suggesting that the uniform measure $\leb$ is a suitable choice for the source measure $Q$. 
	When $\vartheta_n$ exhibits asymmetry, $\Phi_\# P_n$ departs from uniformity, placing more mass on one side of $[0,1]$. 
	A Beta distribution with flexible shape parameters $(\eta, \zeta)$ provides a natural family to model such deviations from uniformity, allowing $Q$ to better reflect the asymmetry of $\vartheta_n$.	
	Following this logic, we estimate $(\eta, \zeta)$ by fitting a Beta distribution to the transformed bootstrap replicates
	\[
	\Phi(\hat\vartheta_n^{(1) \ast}), \ldots, \Phi(\hat\vartheta_n^{(m) \ast}) \sim \mathrm{Beta}(\eta, \zeta),
	\]
	using a method of moments or maximum likelihood. 
	This provides a data-driven adjustment of $Q$ that improves finite-sample coverage while retaining the computational tractability of the discretized source measure.
	
	Given $Q \in \mathcal{Q}_{\beta}$, let the non‐coverage sets $\mathcal{S}^c$ and $\C^c$ be as defined above. 
	To select the number of support points $M$, we minimize the quadratic risk
	\begin{equation}\label{eq:3-quadratic_risk}
		R(M) \coloneqq \expect{ \Big(\sum_{i = 1}^{M}q_i \mathbbm{1}_{i \in \mathcal{S}^c} -
			\sum_{i = 1}^{M} q_i \mathbbm{1}_{i \in \C^c}\Big)^2 } ,
	\end{equation}
	where we choose $\phi(x) = x^2$ as a natural measure of deviation.
	
	The expansion of \eqref{eq:3-quadratic_risk} separates into one deterministic term and two terms that depend on the unknown distribution $P_n$. 
	To estimate these terms, we employ a leave-one-out strategy. 
	Specifically, for each $k = 1, \ldots, n$, let $\hat P_{n, -k}$ denote the empirical measure with the $k$-th observation removed, and let $\hat{\gamma}_{n, -k} \in \Gamma(Q, \hat P_{n, -k})$ be the corresponding  OT plan. Define the leave-one-out non-coverage set
	\[
	\hat{\mathcal{S}}_{-k}^{c} \coloneqq
	\left\{x \in \spt{Q} : \exists y \in \tilde{\mathcal{I}}_{n,Q, -k}^{c} \text{ with }y \in \partial \hat{\varphi}_{n, -k}(x) \right\}.
	\]
	Then construct the approximately unbiased estimators of expectations 
	{ 
	\[
	\hat{A}_M = \sum_{k = 1}^{n}\sum_{i = 1}^{M}\sum_{j = 1}^{M} \frac{q_i q_j}{n} \mathbbm{1}_{i,j \in \hat{\mathcal{S}}_{-k}^{c}}, 
	\ 
	\hat{B}_M = \sum_{k = 1}^{n}\sum_{i = 1}^{M}\sum_{j = 1}^{M} \frac{q_i q_j}{n} \mathbbm{1}_{i \in \hat{\mathcal{S}}_{-k}^{c}} \mathbbm{1}_{j \in \C^c}, 
	\ 
	C_M = \sum_{i = 1}^{M}\sum_{j = 1}^{M} q_i q_j \mathbbm{1}_{i,j \in \C^c} 
	\]
	}
	and find $M$ that minimizes the approximately unbiased risk estimator, i.e.
	\[
	\hat{M} \in \arg\min_{M \vargeq 1} \hat{R}(M) 
	\ \text{  with  }\	\hat{R}(M) = \hat{A}_M - 2 \hat{B}_M + C_M \  .
	\]
	
	We conclude with the remark that
	the OT plan from $Q$ to $\hat{P}_n$ is in general not unique, and the mapping $Q \mapsto \expect{\mathcal{E}_n(Q)^2}$ highly nonlinear, such that the optimization problem
	\[
	Q^\star \in \arg\min_{Q \in \mathcal{Q}} 
	\expect{\mathcal{E}_n(Q)^2}
	\]
	may be ill-posed.
	A natural regularization strategy is to add a penalty to the interval length, with $\lambda > 0$ as penalty parameter, i.e.
	\[
	Q^\star \in \arg\min_{Q \in \mathcal{Q}} 
	\Big\{\expect{\mathcal{E}_n(Q)^2} + \lambda \, \hat{\ell}_n(Q) \Big\},
	\]
similarly to classical regularization methods for ill-posed problems \citep{engl1996regularization} 
	
	
	\section{Finite Sample Validation} \label{sec:montecarlo}
	
	We evaluate the finite-sample performance of the proposed confidence intervals using Monte Carlo simulations and an illustration application. 
	
	\subsection{Simulations} 
	
	To evaluate both coverage and stability, we employ a two-layer design with $100$ replication blocks, each containing $100$ independent samples. We report the empirical coverage, left and right miss-coverage probabilities, and the median interval length, which is robust to extremely wide intervals.
	Specifically, we consider two simple inference problems: (i) estimating the mean when observations are generated from a $\chi^2$ distribution, and (ii) estimating the correlation coefficient between two Gaussian variables. 
	For a sample $X_1,\dots,X_n$ and statistic $\vartheta_n$, let
	$
	\hat\vartheta_n^{(1) \ast}, \dots, \hat\vartheta_n^{(m) \ast},
	$
	denote $m$ bootstrap replicates with empirical distribution $\hat{P}_n$.
	
	For $Q$ a discretized Beta distribution, its parameters $(\eta,\zeta)$ and the number $M$ of support points are selected as described in Section~\ref{subsec:Choices}.
	Here, $M$ is chosen from a grid from $5$ to $100$ in steps of $5$, with optional local refinement.
%
	Then we define the admissible set  
	\[
	\mathcal{C} \coloneqq \{x \in \spt{Q}: F_Q^{-1}(\alpha/2) \varleq x \varleq F_Q^{-1}(1-\alpha/2)\}.
	\]
	For each such $Q$, we compute an  OT plan $\hat{\gamma}=(\hat{\gamma}_{ij})$, where $\hat{\gamma}_{ij}$ denotes the mass transported from $x_i \in \spt{Q}$ to $y_j \in \spt{\hat P_n}$; see \citet{peyre2020computational} for computational background.  
	This plan is then used to construct the confidence interval $\hat{\mathcal{I}}_{n,Q}$. 
	
	For comparison, consider the standard bootstrap procedures for a scalar parameter $\theta_0$:
	\[
	\text{percentile: } \vartheta_n = \hat{\theta}_n, \qquad
	\text{basic: } \vartheta_n = \hat{\theta}_n - \theta_0, \qquad
	\text{studentized: } \vartheta_n = { ( \hat{\theta}_n - \theta_0 ) }{\sigma_{\hat{\theta}_n}^{-1}}.
	\]
	In the bootstrap, $\theta_0$ is replaced by $\hat{\theta}_n$, and $\sigma_{\hat{\theta}_n}$ by a consistent estimate.

	Simulation parameters are set as follows: {nominal level} $\alpha = 0.1$, {number of bootstrap replicates} $m = 1000$, {cost function} $c(x,y) = |x-y|^2/2$, and {sample sizes} $n \in \{5,10,25,50,100\}$.  
	
	To assess the effectiveness of our $Q$-selection procedure, we additionally consider three fixed representative Beta distributions with parameters $(\eta, \zeta) = (2,5), (5,5), (5,2)$, corresponding to left-skewed, symmetric, and right-skewed shapes.

	All simulations are conducted in \texttt{R}, using the \texttt{transport} package for OT computations \citep{R2024transport}, \texttt{boot} for bootstrap resampling \citep{canty2025boot}, and \texttt{fitdistrplus} to estimate the Beta parameters \citep{delignette2015fitdistrplus}. The source measure $Q$ is discretized on an equally spaced grid of $M$ atoms. 
	
	\begin{algorithm}[htb]
		\caption{ OT-Based Confidence Interval}\label{alg:otci}
		\begin{algorithmic}
			\Require $\texttt{data} = X_1, \ldots, X_n$, $\texttt{level} = \alpha$, $\texttt{m} = m$, $\texttt{M} = M$, $\texttt{boot\_func}$: bootstrap function for the statistic, $\texttt{method} \in \{\textbf{perc, basic, stud}\}$
			\State \textbf{1. Bootstrap}
			\State $\{\hat\vartheta_n^{\ast(b)}\}_{b=1}^m \gets \texttt{boot::boot}(\texttt{data}, \texttt{boot\_func}, \texttt{m})$ \Comment{Compute bootstrap replicates}
			\State Transform $\{\vartheta_n^{\ast(b)}\}$ according to $\texttt{method}$
			\State $\hat{P}_n \gets$ empirical probabilities of transformed replicates \Comment{Bootstrap distribution}
			\State \textbf{2. Source Distribution}
			\State $Z \gets \Phi(\hat\vartheta_n^{\ast})$ \Comment{CDF transform of replicates}
			\State $(\eta, \zeta) \gets \texttt{fitdistrplus::fitdist}(Z)$ \Comment{Fit Beta parameters}
			\State $\{x_i, q_i\}_{i=1}^M = Q \gets$ discretized Beta$(\eta, \zeta)$ \Comment{Construct discrete source measure}
			\State \textbf{3. Optimal Transport}
			\State $C_{ij} \gets |x_i - \hat\vartheta_n^{\ast(j)}|^2 / 2$ \Comment{Cost matrix}
			\State $\hat{\bm{\gamma}}_n \gets \texttt{transport::transport}(Q, \hat{P}_n, C)$ \Comment{Compute optimal coupling}
			\State $[i_L, i_U] \gets$ indices in $Q$ corresponding to $\alpha/2$ and $1-\alpha/2$
			\State $\hat{\mathcal{I}}_{n,Q} \gets$ map $[i_L, i_U]$ via $\hat{\bm{\gamma}}_n$ \Comment{Construct OT-based confidence interval}
			\State \Return $\hat{\mathcal{I}}_{n,Q}$
		\end{algorithmic}
	\end{algorithm}
	
	
	From sample $X_1, \ldots, X_n$ of a non-central $\chi^2_1$ distribution with non-centrality parameter $10$ estimate the mean, i.e.\ 
	$ 
	\theta_0 = \mathbb{E}[X] 
	= 11,
	\ 
	\hat{\theta}_n = \frac{1}{n} \sum_{i=1}^n X_i  .
	$ 
	%
 %
	Then compute confidence intervals with the percentile, basic, and studentized bootstrap replicates. 	Performance is assessed via empirical coverage, left and right miss-coverage probabilities, and the median interval length.
	Simulations are conducted for sample sizes $n \in \{5, 10, 25, 50, 100\}$, with $m = 1000$ bootstrap replicates per sample.
%

	\begin{table}[htb]
			\setlength{\tabcolsep}{2pt} 
		\centering
		\caption{Comparison of OT and Bootstrap confidence intervals across methods}
		\label{tab:all_methods}
		\vspace{-0.25cm}
			\centering
			\begin{tabular}{l S S S S S S}
				& \multicolumn{3}{c}{\textbf{ OT}} &
\multicolumn{3}{c}{\textbf{Bootstrap}} \\
\cmidrule(lr){2-4} \cmidrule(lr){5-7}
$n$ & \multicolumn{1}{c}{Coverage} &\multicolumn{1}{c}{Var} & \multicolumn{1}{c}{MSE} & \multicolumn{1}{c}{Coverage} &\multicolumn{1}{c}{Var} & \multicolumn{1}{c}{MSE} \\
				\bottomrule
				\multicolumn{7}{c}{\textbf{Percentile Confidence Intervals for the Mean}}\\
				\toprule

				5  & 0.8803 & 0.0009 & 0.0013 & 0.7727 & 0.0015 & 0.0177 \\
				10 & 0.9067 & 0.0011 & 0.0011 & 0.8391 & 0.0015 & 0.0052 \\
				25 & 0.9025 & 0.0010 & 0.0010 & 0.8721 & 0.0012 & 0.0020 \\
				50 & 0.9129 & 0.0007 & 0.0009 & 0.8897 & 0.0008 & 0.0010 \\
				100 & 0.8993 & 0.0008 & 0.0008 & 0.8987 & 0.0008 & 0.0008 \\
				\bottomrule
				\multicolumn{7}{c}{\textbf{Basic Confidence Intervals for the Mean}}\\
				\toprule
				5  & 0.8192 & 0.0015 & 0.0080 & 0.7648 & 0.0016 & 0.0199 \\
				10 & 0.8530 & 0.0015 & 0.0037 & 0.8340 & 0.0015 & 0.0059 \\
				25 & 0.8869 & 0.0011 & 0.0013 & 0.8702 & 0.0012 & 0.0021 \\
				50 & 0.8951 & 0.0008 & 0.0008 & 0.8847 & 0.0008 & 0.0010 \\
				100 & 0.8973 & 0.0008 & 0.0008 & 0.8967 & 0.0008 & 0.0008 \\
				\bottomrule
				\multicolumn{7}{c}{\textbf{Studentized Confidence Intervals for the Mean}}\\
				\toprule
				5  & 0.9216 & 0.0007 & 0.0012 & 0.8576 & 0.0011 & 0.0029 \\
				10 & 0.9158 & 0.0010 & 0.0012 & 0.8721 & 0.0014 & 0.0021 \\
				25 & 0.9166 & 0.0009 & 0.0012 & 0.8843 & 0.0012 & 0.0014 \\
				50 & 0.9173 & 0.0008 & 0.0011 & 0.8946 & 0.0009 & 0.0009 \\
				100 & 0.9049 & 0.0008 & 0.0009 & 0.9016 & 0.0008 & 0.0008 \\
				\bottomrule
			\end{tabular}
				\vspace{-0.25cm}
	\end{table}
	
	The simulation results, summarized in Tables~\ref{tab:all_methods} and~\ref{tab:06_length_miss_mean}, and Figure~\ref{fig:all_methods_coverage} for $n = \{5,10\}$, compare  OT and bootstrap confidence intervals across the percentile, basic, and studentized methods.
	Across all interval types, our OT procedures clearly outperform the standard bootstrap intervals, particularly for small sample sizes. For $n=5$, the MSE of  OT percentile intervals is roughly 13 times smaller than the bootstrap MSE, reflecting both improved coverage and reduced variability. Similar patterns hold for basic and studentized intervals, indicating that  OT systematically stabilizes interval estimation.
	Surprisingly, while it is known that for bootstrap inference it is worth to use the studentized version, OT-based intervals perform best for the simplest version.  
	Recall that the studentized one requires (and its performance strongly depends on) the estimation of $\sigma_{\hat{\theta}_n}$ which is challenging for small samples. As expected, for increasing $n$ the performance converge, but even for $n=50$ we see a noticeable advantage of the OT-based procedure.  
%
	

Next we draw observations $(X_{i,1}, X_{i,2})$ from a bivariate normal distribution with mean zero, variance one and correlation  $\theta_{0} = \rho = 0.1$. Estimation is performed using the standard Pearson correlation estimator.
As for the mean, the percentile method exhibits the most reliable small-sample performance among the three approaches considered. Its figures are given in Table \ref{Tab-corr}.
We observe that with the  OT approach, we hit the target coverage $1 - \alpha$ for all tested sample sizes. In contrast, the percentile bootstrap approach the nominal level more slowly. Moreover, our method is more stable, see Var and MSE.

\begin{table}[htb]
	\centering  \vspace{-0.25cm}
	\caption{Comparison of coverage probabilities for the percentile method. \label{Tab-corr}}
	\setlength{\tabcolsep}{2pt} 
		\vspace{-0.25cm}
	\begin{tabular}{l S S S S S S}
		& \multicolumn{3}{c}{\textbf{ OT}} &
		\multicolumn{3}{c}{\textbf{Bootstrap}} \\
		\cmidrule(lr){2-4} \cmidrule(lr){5-7}
		$n$ & \multicolumn{1}{c}{Coverage} &\multicolumn{1}{c}{Var} & \multicolumn{1}{c}{MSE} & \multicolumn{1}{c}{Coverage} &\multicolumn{1}{c}{Var} & \multicolumn{1}{c}{MSE} \\
		\midrule
		5  & 0.9038 & 0.0009 & 0.0009 & 0.8693 & 0.0012 & 0.0021 \\
		10 & 0.9032 & 0.0011 & 0.0011 & 0.8440 & 0.0015 & 0.0046 \\
		25 & 0.9259 & 0.0006 & 0.0013 & 0.8736 & 0.0010 & 0.0017 \\
		50 & 0.9017 & 0.0010 & 0.0010 & 0.8800 & 0.0011 & 0.0015 \\
		100 & 0.9022 & 0.0009 & 0.0009 & 0.8864 & 0.0011 & 0.0012 \\
		\bottomrule
	\end{tabular}
\end{table}

	These results show that our finite-sample bounds reflect well the behavior of our intervals. In particular, the  OT stabilizes coverage in small samples while preserving asymptotic validity, exhibiting performance comparable to standard bootstrap methods as $n$ grows.
	
	\subsection{The velocity of Galaxies} 
	
	To further assess the practical performance of our method, we conduct a study based on the \texttt{galaxies} dataset from the \texttt{MASS} package, which contains the velocities (in km/s) of 82 galaxies. It provides a challenging benchmark for statistical inference due to its pronounced multimodality, see Figure \ref{fig:galaxies_density}.
	
	\begin{figure}[htb]
		\centering \vspace{-0.25cm}
		\includegraphics[width=0.75\textwidth,height=5.75cm]{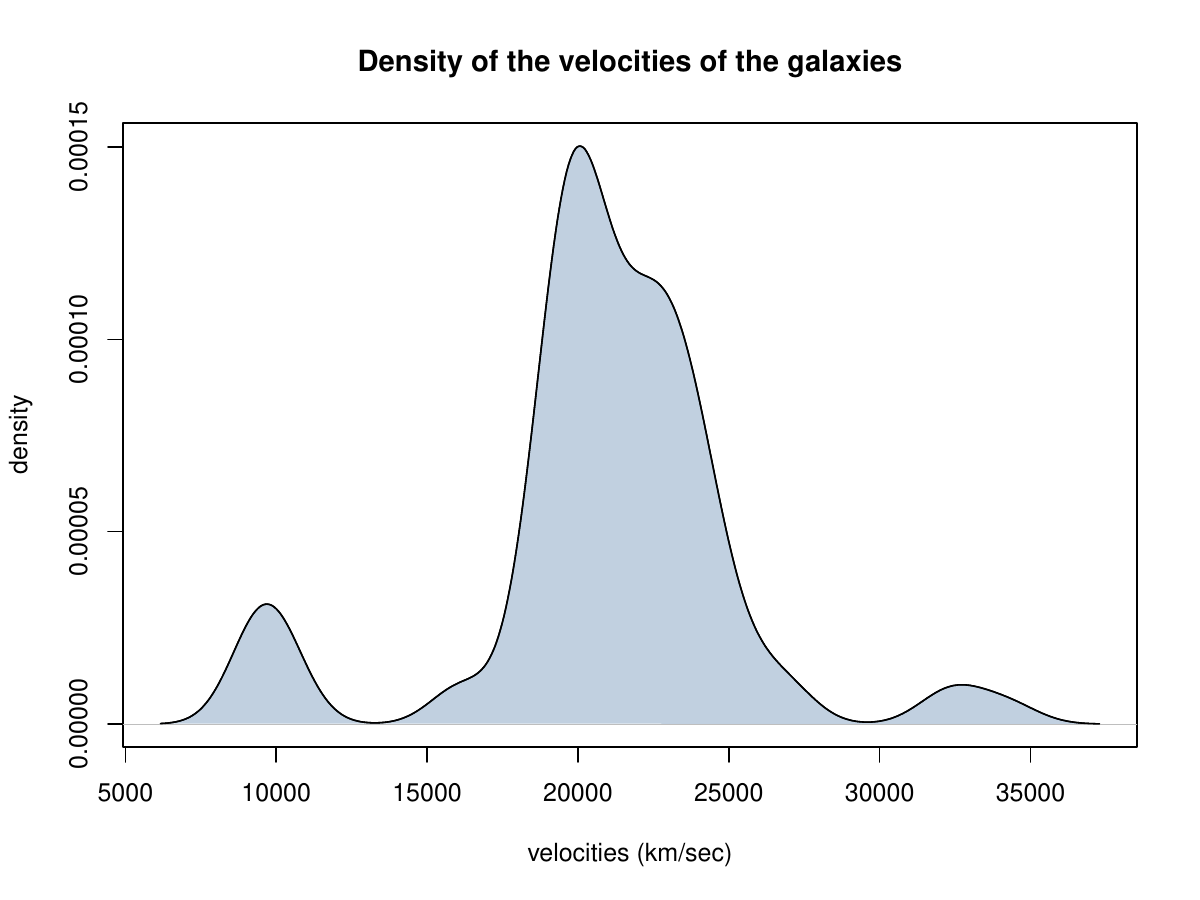}
		\vspace{-0.5cm}
		\caption{Estimated density of the galaxies velocity dataset.}
		\label{fig:galaxies_density}
	\end{figure}
	
	We model the distribution of the velocity $X$ by a Gaussian mixture with three components:
	\[
	X \sim \sum_{i=1}^{3} w_i \mathcal{N}(\mu_i, \sigma_i^2),
	\quad \sum_{i=1}^{3} w_i = 1,
	\]
	and consider as parameter of interest the overall mean of the mixture,
	\[
	\theta_0 = \sum_{i=1}^{3} w_i \mu_i,
	\qquad
	\hat\theta_n = \frac{1}{n} \sum_{i=1}^{n} X_i
	\]
	but could equally well consider the problem of constructing confidence intervals for each of the three groups. 
	More specifically, in this illustration, we generate data from a fixed Gaussian mixture model fitted to the full galaxies dataset. As in previous experiments, we adopt a two-layer simulation design consisting of 100 Monte Carlo blocks, each containing 100 independent samples of size $n$ drawn from the fitted mixture model.
	Similarly to above, we consider a nominal confidence level of $\alpha = 0.1$,
	$m=1000$ bootstrap replicates, sample sizes $n \in \{5, 10\}$, and the
	admissible set
		\[
		\mathcal{C} = \big\{ 
		x \in \mathrm{supp}(Q) \ :\ F_Q^{-1}(\alpha/2) \le x \le F_Q^{-1}(1 - \alpha/2)
		\big\} .
		\]
	We compare the performance of our approach with standard bootstrap procedures, namely the percentile, basic, and studentized methods.
	
	\begin{table}[htb]
		\centering   	\setlength{\tabcolsep}{2pt} 
		\vspace{-0.25cm}
	\caption{Coverage and Length of Confidence for the Average Velocity of Galaxies} \label{tab-appl-galaxies}			
		\begin{tabular}{l S S S S S S}
			& \multicolumn{3}{c}{\textbf{ OT}} &
\multicolumn{3}{c}{\textbf{Bootstrap}} \\
\cmidrule(lr){2-4} \cmidrule(lr){5-7}
$n$ & {Coverage} & {Var} & {MSE} & {Coverage} & {Var} & {MSE} \\
\midrule
			\multicolumn{7}{c}{\textbf{Percentile Confidence Intervals}}\\
			\toprule
			5  & 0.8680 & 0.0013 & 0.0023 & 0.7514 & 0.0017 & 0.0238 \\
			10 & 0.9061 & 0.0009 & 0.0009 & 0.7912 & 0.0011 & 0.1296 \\
			\bottomrule
%
			\multicolumn{7}{c}{\textbf{Basic Confidence Intervals}}\\			\toprule
			5  & 0.8522 & 0.0013 & 0.0036 & 0.7972 & 0.0016 & 0.0122 \\
			10 & 0.8643 & 0.0009 & 0.0022 & 0.8412 & 0.0012 & 0.0047 \\
			\bottomrule
%
			\multicolumn{7}{c}{\textbf{Studentized Confidence Intervals}}\\
			\toprule
			5  & 0.9222 & 0.0006 & 0.0011 & 0.8560 & 0.0012 & 0.0031 \\
			10 & 0.8343 & 0.0012 & 0.0055 & 0.8281 & 0.0012 & 0.0064 \\
			\bottomrule
%
			\multicolumn{7}{c}{\textbf{Interval Length}}\\
			\toprule
			& {Percentile} & {Basic} & {Studentized} & {Percentile} & {Basic} & {Studentized} \\			\midrule
			5  & 5900.338 & 5133.451 & 15262.498 & 4409.446 & 4405.495 & 6279.923 \\
			10 & 6054.043 & 4451.817 & 6553.114 & 4199.520 & 4204.053 & 4945.214 \\
			\bottomrule
		\end{tabular}
	\end{table}
	
	The results show that the proposed  OT-based procedure generally achieves improved coverage relative to classical bootstrap methods, while exhibiting substantially lower variance and mean squared error in most configurations. In particular, gains are more pronounced for small sample sizes ($n=5$), where standard bootstrap methods tend to under-cover the true parameter. The studentized bootstrap improves coverage in some cases but displays increased instability, reflected in larger interval lengths and higher variability, especially in the smallest sample regime. Overall, the proposed approach provides a more stable trade-off between coverage accuracy and interval efficiency in this multimodal setting.

	\section{Conclusion and Discussion} \label{sec:concl}
	
	We have shown that confidence intervals constructed via  OT can exhibit substantially improved coverage properties. Our theoretical validation relied on three components:  
	(1) {Large-sample consistency},
	(2) {Finite-sample guarantees}, and
	(3) {Length analysis}.
	Although fully data-driven, our method is easy to implement, very fast and has proven to be amazingly reliable, outperforming 'direct' bootstrap intervals by far.   
	
	A natural extension 
	is the construction of confidence sets in $\mathbb{R}^d$ for $d \vargeq 2$. Our analysis leveraged the special structure of one-dimensional  OT, where monotonicity yields substantial simplifications. In higher dimensions, theoretical guarantees require different tools, and the performance may depend more sensitively on the choice of source distribution. In particular, if $Q$ is absolutely continuous, the primary object is the  OT map, whereas if $Q$ is discrete, one must work instead with an optimal coupling. A comprehensive multivariate theory must therefore accommodate both regimes.
	
%
%
	
	Another extension could be to non-Euclidean settings, notably to finite-dimensional Riemann manifolds. 	The work of \citet{mccann2001polar} establishes existence and regularity of optimal maps for costs induced by the squared Riemann distance.
	This provides a foundation for  OT–based inference in curved state spaces, including confidence regions for parameters that lie on manifolds (e.g., directions, covariance structures, or spectral data).
	
	In this work, we introduced a data-driven mechanism for selecting source distributions, implemented through discretized Beta laws. This choice yields a tractable parametrization but is not canonical, and questions of optimality remain open.

	\bibliographystyle{apalike}
	\bibliography{bibliography}

	\appendix
	
\section{Proofs} 	\label{sec-append-proofs}

 \subsection{Proof of Lemma~\ref{lem:empirical_extreme_deviation}}	

	\begin{proof}
		Fix an index $i \in \{1, \ldots, M \}$ and the source measure $Q \in \mathcal{Q}$.
		From the monotone rearrangement of optimal coupling on $\Real$, the mass $q_i = Q(\{x_i\})$ is transported to the set of target points that lies in the interval $[ F^{-1}(V_{i - 1}), F^{-1}(V_i)]$, hence
		\begin{eqnarray*}
			& u_i^{-} = F^{-1}(V_{i - 1}),		\quad
			\hat{u}_i^{-} = \hat{F}_n^{-1}(V_{i - 1})  , \quad
			u_i^{+} =  F^{-1}(V_i),  \quad
			\hat{u}_i^{+} = \hat{F}_n^{-1}(V_i)  . &
		\end{eqnarray*}
		Taking the differences and the absolute values yields
		\[
		\lvert u_i^{-} - \hat{u}_i^{-} \rvert =	\lvert 
		F^{-1}(V_{i - 1}) - \hat{F}_n^{-1}(V_{i - 1})		\rvert \varleq
		\sup_{0 \varleq u \varleq 1} \lvert F^{-1}(u) - \hat{F}_n^{-1}(u)\rvert,
		\]
		and similarly for $\lvert u_i^{+} - \hat{u}_i^{+} \rvert$.
		As $i$ and $Q$ were arbitrary, the claim follows immediately.
	\end{proof}

	\subsection{Proof of Lemma~\ref{lem:3-upper_bound_probability}}
	
	\begin{proof}
		Conditional on data $X_1, \ldots, X_n$,
		fix an index $i$, a source probability measure $Q \in \mathcal{Q}$, and an admissible set $\C$.
		By construction, the event $\{x_i \in \mathcal{S}^c \, \vert \ X_1,, \ldots, X_n\}$ implies that at least one of the following inequalities fails:
		$
		\hat{u}_a \varleq u_i^{-} \
		\text{ or } \
		\hat{u}_b \vargeq u_i^{+}.
		$
		Equivalently, for $\hat{u}_i^{-}$, $\hat{u}_i^{+}$ defined as ${u}_i^{-}$, ${u}_i^{+}$ but with $\hat \varphi_i$, we have 
		$\{x_i \in \mathcal{S}^c \, \vert \ X_1,, \ldots, X_n\} \implies \{\max \big( \hat{u}_a - u_i^{-}, u_i^{+} - \hat{u}_b\big) > 0 \, \vert \ X_1,, \ldots, X_n\}$.
		We proceed now by contradiction. Suppose that
		\[
		\Big\{\max \big(
		\lvert u_a - \hat{u}_a \rvert, \lvert u_b - \hat{u}_b \rvert 
		\big) \varleq \frac{g_i}{2} \, \vert \ X_1,, \ldots, X_n\Big\},
		\]
		using the definition of $g_i$, we have $u_i^{-} - u_a \vargeq g_i$ and $u_b - u_i^{+} \vargeq g_i$.
		Hence,
		\[
		\hat{u}_a \varleq u_a + \frac{g_i}{2} \varleq u_i^{-} - \frac{g_i}{2} \varleq u_i^{-},
		\qquad
		\hat{u}_b \vargeq u_b - \frac{g_i}{2} \vargeq u_i^{+} + \frac{g_i}{2} \vargeq u_i^{+},
		\]
		and neither exclusion conditions can hold, therefore $\{x_i \not\in \mathcal{S}^c \, \vert \ X_1,, \ldots, X_n\}$.
		Taking the complement yields
		\[
		\Big\lbrace x_i \in \mathcal{S}^c \, \vert \ X_1,, \ldots, X_n \Big\rbrace \subseteq
		\Big\lbrace 
		\max \lbrace 
		\lvert u_a - \hat{u}_a \rvert, 
		\lvert u_b - \hat{u}_b \rvert 
		\rbrace > \frac{g_i}{2}
		\, \vert \ X_1,, \ldots, X_n
		\Big\rbrace   .
		\]
		By Lemma~\ref{lem:empirical_extreme_deviation}, the deviations are bounded from above by $\sup_{0 \varleq u \varleq 1} \lvert F_n^{-1}(u) - \hat{F}_n^{-1}(u) \rvert $.
		Taking probabilities and expectation with respect to the data gives the claimed bound unconditionally.
		From the monotonicity of the subdifferential, if $x_i < c_a$ or $x_i > c_b$ then $g_i < 0$, and the probability equals 1.
	\end{proof}
	


	\subsection{Proof of Corollary~\ref{crll:expectation_noncoverage_II}}
	
	\begin{proof}
		Define $S \coloneqq \alpha + \sum_{i \in \C} q_i\big(2\exp(-\frac{m f_{\min}^2 g_i^2 }{2}) + A(g_i) \big)$ and $\sigma^2 \coloneqq \sum_{i = 1}^{M}q_i^2$.
		From Theorem~\ref{th:3-upper_bound_noncoverage_II}, we have
		\[
		\pr{\overline{\pi}_n(Q) \vargeq S + \sqrt{2\ln(1/\delta) \sigma^2}} \varleq \delta.
		\]
		Let $t = S + \sqrt{2\ln(1/\delta) \sigma^2}$, solving for $\delta$ gives $\delta = \exp(-\frac{(t - S)^2}{2\sigma^2})$.
		Hence, for all $t \vargeq S$,
		\[
	\pr{\overline{\pi}_n(Q) \vargeq t} \varleq \exp\Big(-\frac{(t - S)^2}{2\sigma^2}\Big) .
		\]
		Since $0 \varleq \overline{\pi}_n(Q) \varleq 1$ a.s., the expectation can be written using the tail integral formula
		\[
		\expect{\overline{\pi}_n(Q)} = 
		\int_{0}^{1} \pr{\overline{\pi}_n(Q) \vargeq t} dt \varleq
		\int_{0}^{1} 1 \wedge \exp\Big(-\frac{(t - S)^2}{2\sigma^2}\Big) dt \ .
		\]
		Setting $t_0 = S$ where the exponential term equals 1, we see
		\begin{align*}
			\int_{0}^{1} 1 \wedge \exp\Big(-\frac{(t - S)^2}{2\sigma^2}\Big) dt &=
			\int_0^S dt + \int_S^1 \exp\Big(-\frac{(t - S)^2}{2\sigma^2}\Big) dt \\
			&= S + \sqrt{\frac{\pi\sigma^2}{2}} \text{erf}\Big(\frac{1 - S}{\sqrt{2 \sigma^2}}\Big),
		\end{align*}
		with error function $\text{erf}(x)$.
		As $\text{erf}(x) \varleq 1$ for all $x \in \Real$, we obtain the claimed bound. 
	\end{proof}

	\subsection{Proof of Lemma~\ref{lem:expectation_sup-quantile}}
	
	\begin{proof}
		From the assumption $f(x) \vargeq f_{\min} > 0$ and the proof of Theorem~\ref{th:3-upper_bound_noncoverage_II} we have
		\[
		\sup_{0 \varleq u \varleq 1} \lvert F_n^{-1}(u) - \hat{F}_n^{-1}(u) \rvert \varleq
		\frac{1}{f_{\min}} \sup_{x \in \Real} \lvert F_n(x) - \hat{F}_n(x) \rvert.
		\]
		Since this quantity is positive, the tail integral representation of the expectation and the DKW-M inequality gives:
		\begin{align*}
		&	\expect{\sup_{0 \varleq u \varleq 1} \lvert F_n^{-1}(u) - \hat{F}_n^{-1}(u) \rvert} =
			\int_{0}^{\infty} \pr{\supquant{u \in [0, 1]} \vargeq t} dt \\
		&	\varleq \int_{0}^{\infty} \pr{\supquant{x \in \Real} \vargeq f_{\min} t} dt 
			\varleq \int_{0}^{\infty} 1 \wedge 2\exp\Big(-2 n f_{\min}^2 t^2\Big) dt
		\end{align*}
		Setting $t_0 = \sqrt{\frac{\ln{2}}{2 n f_{\min}^2}}$ where the exponential term is 1, evaluating the integral gives
		\begin{eqnarray*}
		\int_{0}^{\infty} 1 \wedge 2\exp\Big(-2 n f_{\min}^2 t^2\Big) dt &=&
		\int_{0}^{t_0} dt + \int_{t_0}^{\infty} 2\exp\Big(-2 n f_{\min}^2 t^2\Big) dt 
		\\ &=&
		\frac{1}{\sqrt{2n}f_{\min}}\big(\sqrt{\ln{2}} + \sqrt{\pi}\ \text{erfc}(\sqrt{\ln{2}})\big) \ ,
		\end{eqnarray*}
		since $\text{erfc}(\sqrt{\ln{2}}) < 1$, this  yields the claimed bound.
	\end{proof}

	\subsection{Proof of Corollary~\ref{crl:expectation_length_II}}

	\begin{proof}
		Let $S \coloneqq C / \sqrt{m} + K \sqrt{\ln(2/\delta) / 2}$.
		From Corollary~\ref{th:3-2_bound_length_II} we had
		\[
		\pr{\lvert \ell_n(Q) - \ell(Q) \rvert > 2 r_n^{-1}S} \varleq
		\delta.
		\]
		Setting $t = 2 r_n^{-1}S$ gives $\delta = 2 \exp(-2\frac{(t \ r_n / 2 - C/\sqrt{m})^2}{K^2} )$.
		Hence, for all $t > 0$
		\[
		\pr{\lvert \ell_n(Q) - \ell(Q) \rvert > t} \varleq 2 \exp\Big(-2\frac{(t \ r_n / 2 - C/\sqrt{m})^2}{K^2}\Big).     
		\]
		From the tail integral representation of the expectation we have
		\begin{align*}
			\expect{\lvert \ell_n(Q) - \ell(Q) \rvert} &= \int_{0}^{K} \pr{\lvert \ell_n(Q) - \ell(Q) \rvert > t} dt \\
			&\varleq \int_0^K 1 \wedge 2 \exp\Big(-2\frac{(t \ r_n / 2 - C/\sqrt{m})^2}{K^2}\Big) dt \ .
		\end{align*}
		Setting $t_0 = 2r_n^{-1}\frac{C}{\sqrt{m}}$ where the exponential term equals 1,  solving this integral gives
		\begin{align*}
			\int_0^K 1 \wedge 2 \exp\Big(-2\frac{(t \ r_n / 2 - C/\sqrt{m})^2}{K^2}\Big) dt &=
			\int_0^{t_0} dt + \int_{t_0}^K 2 \exp\Big(-2\frac{(t \ r_n / 2 - C/\sqrt{m})^2}{K^2}\Big) dt \\
			&= t_0+ 2r_n^{-1}K\sqrt{\frac{\pi}{2}}\Big(\text{erf}\big(\frac{\sqrt{2}}{K}(K r_n / 2 - C/\sqrt{m})\big)\Big),
		\end{align*}
		Since $\text{erf}$ is bounded above by 1, the expression simplifies giving the claimed expectation.
	\end{proof}

\section{Confidence Interval for $c(x,y) = h(x-y)$}\label{app:convex_cost_confidence_interval}
	
	In this section, we analyze  OT-based confidence interval for more general cost functions.
	Let $c(x,y) = h(x - y)$, where $h(z)$ is a non-negative and strictly convex function.
	Analogous to the quadratic case, an optimal coupling for $h(x - y)$ is concentrated on a $c$-cyclically monotone set.
	A set $S \subset \Real \times \Real$ is said to be $c$-cyclically monotone if for any $k \in \mathbb{N}$ and any finite family $(x_1, y_1), \ldots, (x_k, y_k)$ of points in $S$:
	\[
	\sum_{i = 1}^{k} c(x_i, y_i) \varleq \sum_{i = 1}^{k} c(x_i, y_{i + 1}),
	\]
	with $y_{k + 1} = y_1$.
	By a generalized version of Rockafellar's theorem, a $c$-cyclically monotone set is contained on the graph of the $c$-superdifferential of a $c$-concave function $\varphi$.
	Let $\varphi$ be a candidate for the dual problem, then from the constraint $\varphi(x) + \psi(y) \varleq h(x - y)$, we can set the $c$-concave function
	\[
	\psi(y) = \varphi^c(y) \coloneqq 
	\inf_{x \in \Real} \Big(h(x-y) - \varphi(x) \Big),
	\]
	as the $c$-transform of $\varphi$.
	By symmetry $\varphi(x) = \psi^c(y) = \varphi^{cc}(x)$ is a $c$-concave function.
	The function $\varphi$ and $\varphi^c$ are said to be $c$-conjugate, and the $c$-superdifferential of the $c$-conjugate function $\psi$ at $x \in \Real$ is defined as
	\[
	\partial^{c} \psi(x) \coloneqq
	\Big\{y \in \Real : \varphi(x) + \psi(y) = h(x - y)\Big\}.
	\]
	We denote an optimal pair by $(\varphi, \psi)$.
	Let $c(x,y) = h(x-y)$, then $y \in \partial^c \psi(x)$ implies $\partial \psi(x) \subset \partial h(x - y)$; see~\citep[Lemma~3.1]{gangbo1996geometry}.
	Equivalently -- and if $h$ is differentiable -- we have
	\[
	\partial^{c} \psi(x) =
	\Big\lbrack
	x - (h')^{-1}\!\big(D^{+}\psi(x)\big), \
	x - (h')^{-1}\!\big(D^{-}\psi(x)\big)
	\Big\rbrack.
	\]
	If $(h')^{-1}$ is undefined, we set $(h')^{-1} \coloneqq (h^{\ast})'$ through the Legendre-Fenchel transform instead (cf.\ \citep[p.121]{gangbo1996geometry}).
	
	Consider, as before, the  OT from $Q$ to $P_n$ but for cost $c(x,y) = h(x-y)$, and $c$-confidence intervals being defined as
	\begin{equation}
		\mathcal{I}_{n,Q,c} \coloneqq
		\Big\lbrack
		\hat{\theta}_n - r_n^{-1}\big\{c_a - (h')^{-1}\!\big(D^{-}\psi_n(c_a)\big)\big\}, \
		\hat{\theta}_n - r_n^{-1}\big\{c_b - (h')^{-1}\!\big(D^{+}\psi_n(c_b)\big)\big\}
		\Big\rbrack ,
	\end{equation}
	where $\psi_n$ is a $c$-concave potential for the dual problem.
	
	Similarly, empirical $c$-confidence intervals can be defined from the  OT from $Q$ to $\hat{P}_n$ by the plugin estimator
	\[
	\hat{\mathcal{I}}_{n,Q,c} \coloneqq
	\Big\lbrack
	\hat{\theta}_n - r_n^{-1}\big\{c_a - (h')^{-1}\!\big(D^{-}\hat{\psi}_n(c_a)\big)\big\}, \
	\hat{\theta}_n - r_n^{-1}\big\{c_b - (h')^{-1}\!\big(D^{+}\hat{\psi}_n(c_b)\big)\big\}
	\Big\rbrack  ,
	\]
	where $\hat{\psi}_n$ is a $c$-concave potential for the dual problem.
	
	Most of the results from the article remain valid under some modifications like changing the subdifferential of the convex potentials to $c$-superdifferential of the $c$-concave potential.
	Hence, our results can easily be adapted to the present setting.
	In particular, Lemma~\ref{lem:empirical_extreme_deviation} and Lemma~\ref{lem:3-upper_bound_probability} are still valid under appropriate modifications.
	
		By similar arguments, we can extend the above construction to cost $c(x,y) = l(\lvert x - y \rvert)$, where $l: \lbrack 0, \infty) \rightarrow \lbrack 0, \infty]$ is a strictly concave function.
		From an economic point of view, concave functions form an interesting class of cost: transporting a large amount of mass incurs less total cost per unit than transporting the same mass in smaller, separated shipments.
		However, the  OT is no longer monotone but rather antitone, and if $Q \ll \text{leb}$ the unique  OT map pushing $Q$ to $P$ is defined by
		\begin{equation}
			T_Q(x) \coloneqq F_P^{-1}\big(1 - F_Q(x) \big),
			\quad
			\text{for $Q$-almost all } x \in \spt{Q}.
		\end{equation}
		Hence, special care must be taken when using $c(x,y) = l(\lvert x - y \rvert)$.

	\section{Additional Simulation Results}\label{app:simulation}

	In this section, we present additional results from the Monte Carlo simulations. In particular, we report the distributions of coverage, and the median interval length for all considered methods, along with the left- and right-tail miss-coverage probabilities, providing a comprehensive view of both accuracy and stability.

	\subsection{Simulation Results for estimating the Mean}
	
First we show Figure \ref{fig:all_methods_coverage_n25-50} which completes Table \ref{tab:all_methods} by showing the distribution of the empirical coverages we found in our simulations. It underlines very well what was already indicated in the table: the OT-based method works very well and clearly outperforms the Bootstrap methods not just for tiny samples but also when $n=50$.  
	
\begin{figure}[htb]
	\centering  \vspace{-0.6cm}
	\includegraphics[width=\linewidth,height=6.15cm]{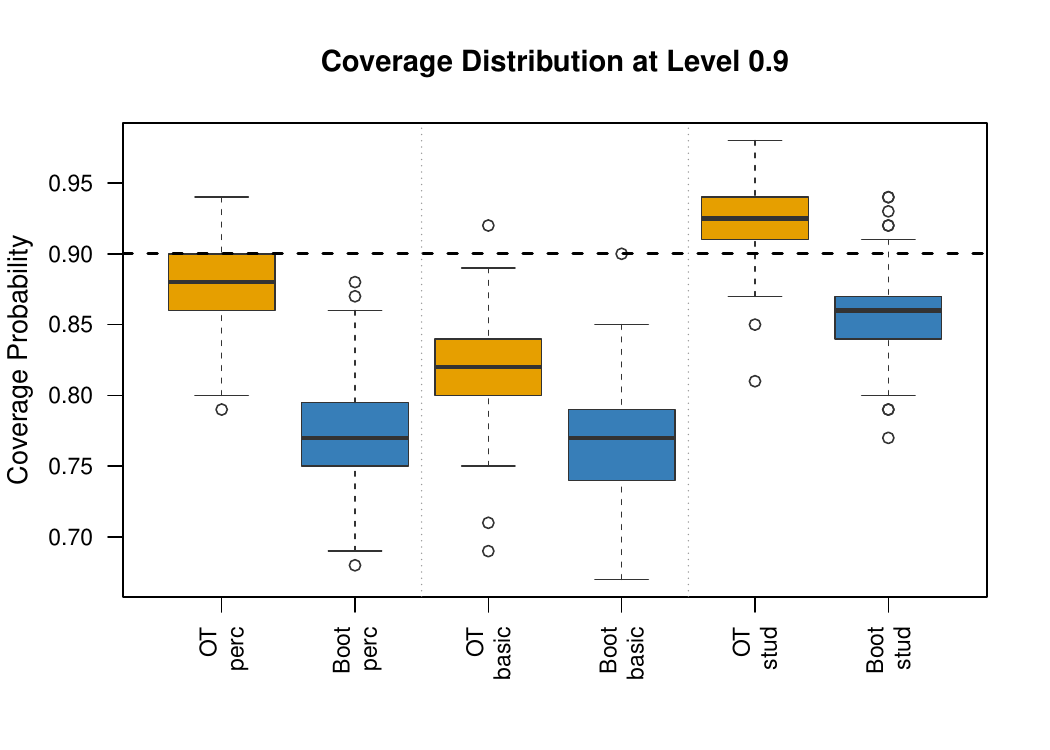}
	\\[-0.6cm]
	\includegraphics[width=\linewidth,height=6.15cm]{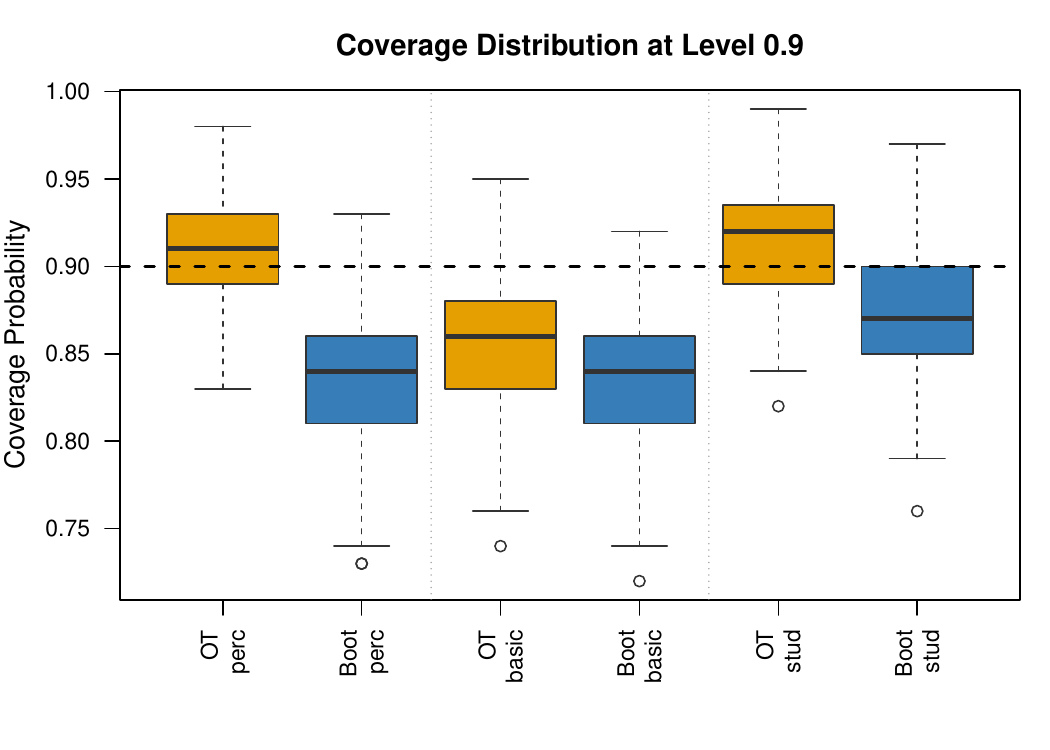}
	\\[-0.6cm]
	\includegraphics[width=\linewidth,height=6.15cm]{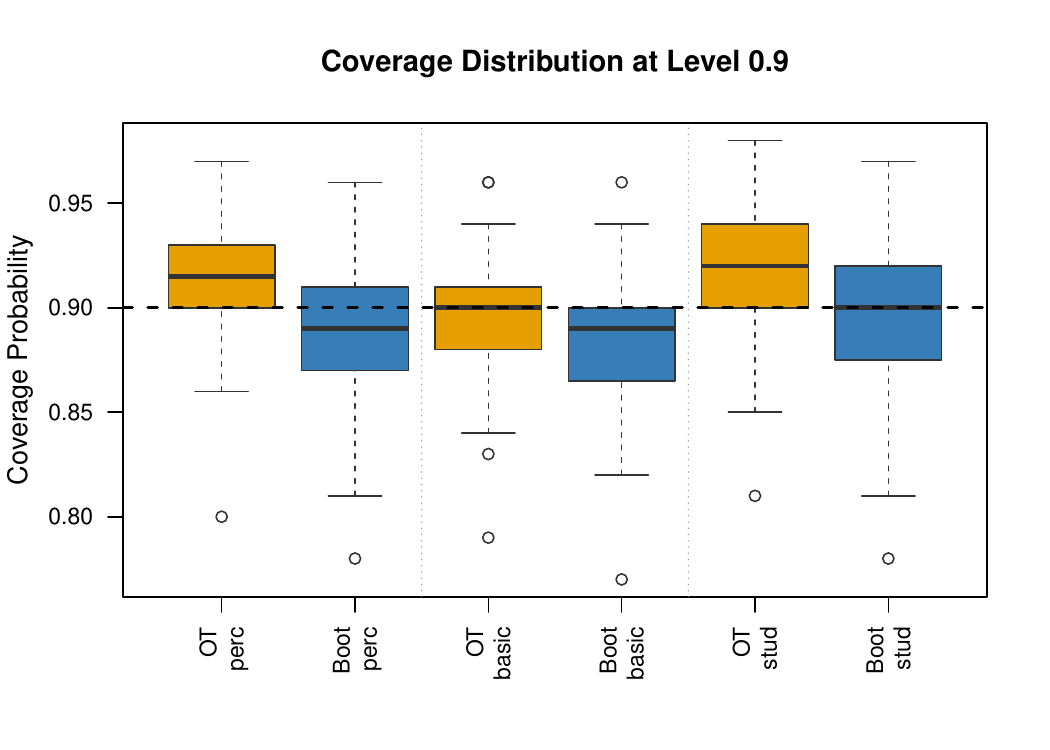}
	\vspace{-1.1cm}
	\caption{Comparison of the coverage probabilities across different methods and sample size ($5$ upper, $10$ center, $50$ lower panel) for confidence intervals of the mean.}
	\label{fig:all_methods_coverage_n25-50}
	\vspace{-0.25cm}
\end{figure}	
	
Table~\ref{tab:06_length_miss_mean} reports the median length and the empirical errors at the left- and right- side of the intervals,
which for symmetric distributions and intervals would equal $\alpha/2 = 0.05$ for our 90\% nominal level.  
These figures reveal where coverage errors occur: a method may achieve correct overall coverage while remaining systematically biased to one side.
	
	\begin{table}[htb]
		\centering
		\caption{Comparison of confidence intervals across different methods}
		\label{tab:06_length_miss_mean}  \vspace{-0.25cm}
			\begin{tabular}{l S S S S S S}
				& \multicolumn{3}{c}{\textbf{ OT}} &
\multicolumn{3}{c}{\textbf{Bootstrap}} \\
\cmidrule(lr){2-4} \cmidrule(lr){5-7}
$n$ & \multicolumn{1}{c}{Percentile} & \multicolumn{1}{c}{Basic} & \multicolumn{1}{c}{Studentized} & \multicolumn{1}{c}{Percentile} & \multicolumn{1}{c}{Basic} & \multicolumn{1}{c}{Studentized} \\
\midrule
				\multicolumn{7}{c}{\textbf{Interval Length}}\\
				\toprule
				5 & 11.3173 & 8.5891 & 17.6637 & 7.4529 & 7.5582 & 10.2471 \\
				10 & 8.5293 & 6.3996 & 8.0181 & 6.0260 & 6.0260 & 6.6384 \\
				25 & 4.4981 & 4.3349 & 4.6915 & 4.0754 & 4.0726 & 4.2175 \\
				50 & 3.2078 & 3.0207 & 3.2483 & 2.9556 & 2.9556 & 3.0041 \\
				100 & 2.1172 & 2.1168 & 2.1610 & 2.1129 & 2.1122 & 2.1308 \\
				\bottomrule
			\end{tabular}
		
		\vspace{0.5em}
		
			\setlength{\tabcolsep}{2pt} 
			\resizebox{\textwidth}{!}{
				\begin{tabular}{l c c c c c c}
					\multicolumn{7}{c}{\textbf{Missing Probabilities}} \\
					\midrule
					5 &
					(0.0446, 0.0751) &
					(0.0490, 0.1318) &
					(0.0228, 0.0556) &
					(0.0720, 0.1553) &
					(0.0692, 0.1660) &
					(0.0426, 0.0998) \\
					10 & 
					(0.0339, 0.0594) &
					(0.0388, 0.1082) &
					(0.0294, 0.0548) &
					(0.0524, 0.1085) &
					(0.0473, 0.1187) &
					(0.0414, 0.0865) \\
					25 & 
					(0.0350, 0.0625) &
					(0.0317, 0.0814) &
					(0.0369, 0.0465) &
					(0.0467, 0.0812) &
					(0.0413, 0.0885) &
					(0.0434, 0.0723) \\
					50 & 
					(0.0344, 0.0527) &
					(0.0345, 0.0704) &
					(0.0409, 0.0418) &
					(0.0438, 0.0665) &
					(0.0399, 0.0754) &
					(0.0436, 0.0618) \\
					100 &
					(0.0434, 0.053) &
					(0.0397, 0.0630) &
					(0.0466, 0.0485) &
					(0.0435, 0.0578) &
					(0.0402, 0.0631) &
					(0.0427, 0.0557) \\				
					\bottomrule
				\end{tabular}   }
\vspace{-0.25cm}
	\end{table}
	
	Across all sample sizes,  OT intervals exhibit markedly more balanced errors than the bootstrap intervals. For $n=5$, e.g.,  OT percentile intervals give $(0.0446,0.0751)$, whereas the bootstrap percentile produces the much more asymmetric $(0.0720,0.1553)$, reflecting undercoverage and instability in the bootstrap tail. 
%
	As $n$ increases, all methods seem to converge towards  $(0.05,0.05)$, but  OT intervals converge more rapidly and remain better balanced. This confirms that the advantage of  OT procedures lies not only in improved total coverage but also in substantially more symmetric tail behavior, reducing the influence of skewness and extreme bootstrap replicates in small samples.
	
	We also report coverage results for three fixed (i.e., not data-adaptively chosen) pairs of Beta parameters, namely $(\eta,\zeta)\in\{(2,5),\,(5,2),\,(5,5)\}$, but only  
for the percentile-based intervals.
Table~\ref{tab:06_beta_mean} presents coverage results for  OT percentile intervals using three fixed Beta parameters $(\eta, \zeta) \in \{(2,5), (5,2), (5,5)\}$, compared with the data-driven approach.  
	
	\begin{table}[htb]
		\centering  
		\caption{Comparison of Confidence Interval for fixed Beta parameters}
		\label{tab:06_beta_mean}  \vspace{-0.25cm}
				\setlength{\tabcolsep}{2pt} 
	\resizebox{\textwidth}{!}{
			\begin{tabular}{l S S S S S S S S S}
& \multicolumn{3}{c}{$(\eta, \zeta) = (2,5)$} 
& \multicolumn{3}{c}{$(\eta, \zeta) = (5,2)$} 
& \multicolumn{3}{c}{$(\eta, \zeta) = (5,5)$} \\ 
				\toprule
				n & \multicolumn{1}{c}{Coverage} & \multicolumn{1}{c}{Var} & \multicolumn{1}{c}{MSE} & \multicolumn{1}{c}{Coverage} & \multicolumn{1}{c}{Var} & \multicolumn{1}{c}{MSE} & \multicolumn{1}{c}{Coverage} & \multicolumn{1}{c}{Var} & \multicolumn{1}{c}{MSE} \\
				\midrule
5   & .8491 & .0012 & .0038 & .8805 & .0008 & .0012 & .8798 & .0009 & .0013 \\
10  & .9050 & .0009 & .0009 & .9405 & .0006 & .0022 & .9125 & .0009 & .0010 \\
25  & .9175 & .0008 & .0011 & .9338 & .0007 & .0018 & .9093 & .0010 & .0011 \\
50  & .9233 & .0006 & .0012 & .9255 & .0008 & .0014 & .9028 & .0009 & .0009 \\
100 & .9003 & .0009 & .0009 & .8987 & .0008 & .0008 & .9013 & .0008 & .0008 \\
				\bottomrule
			\end{tabular} }
		\vspace{-0.25cm}
	\end{table}
	
The data-driven method consistently achieves coverage near the nominal 0.90 target across all sample sizes, providing stable and reliable intervals. While fixed-Beta intervals may occasionally lie closer to 0.90 than the data-adaptively one, none of the fixed-Beta configuration maintains coverage near 0.90 across all sample sizes.  
This highlight the advantage of the data-driven  OT percentile method: it consistently targets the nominal level, providing balanced coverage without the risk of under or over-coverage.

	\begin{figure}[htb]
		\centering  \vspace{-0.35cm}
		\begin{subfigure}[t]{0.85\textwidth}
			\centering
			\caption{Distribution of $\widehat{M}$ for $n = 5$}
			\includegraphics[width=0.95\linewidth,height=10.0cm]{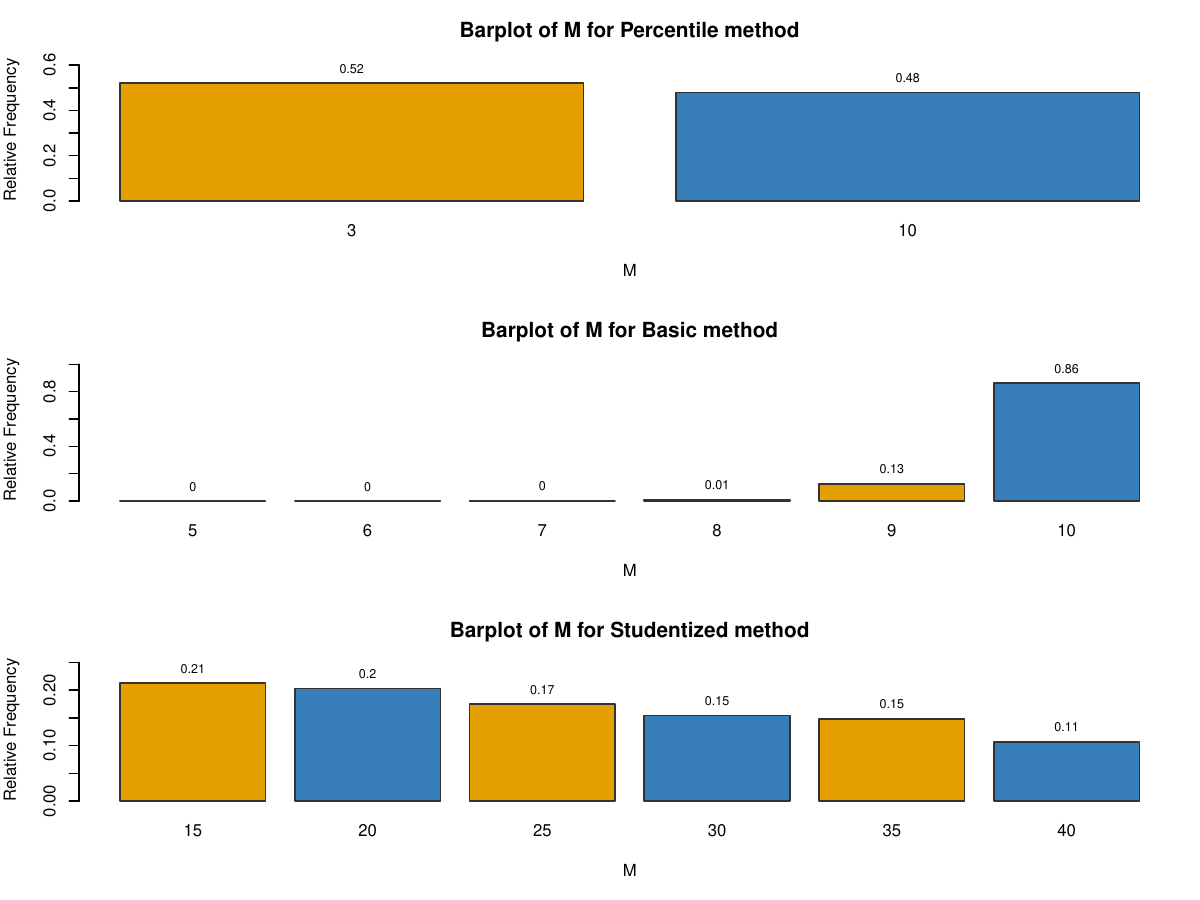}
			\vspace{-0.25cm}
		\end{subfigure}
		\vspace{0.35cm}
		\begin{subfigure}[t]{0.85\textwidth}
			\centering
			\caption{Distribution of $\widehat{M}$ for $n = 10$}
			\includegraphics[width=0.95\linewidth,height=10.0cm]{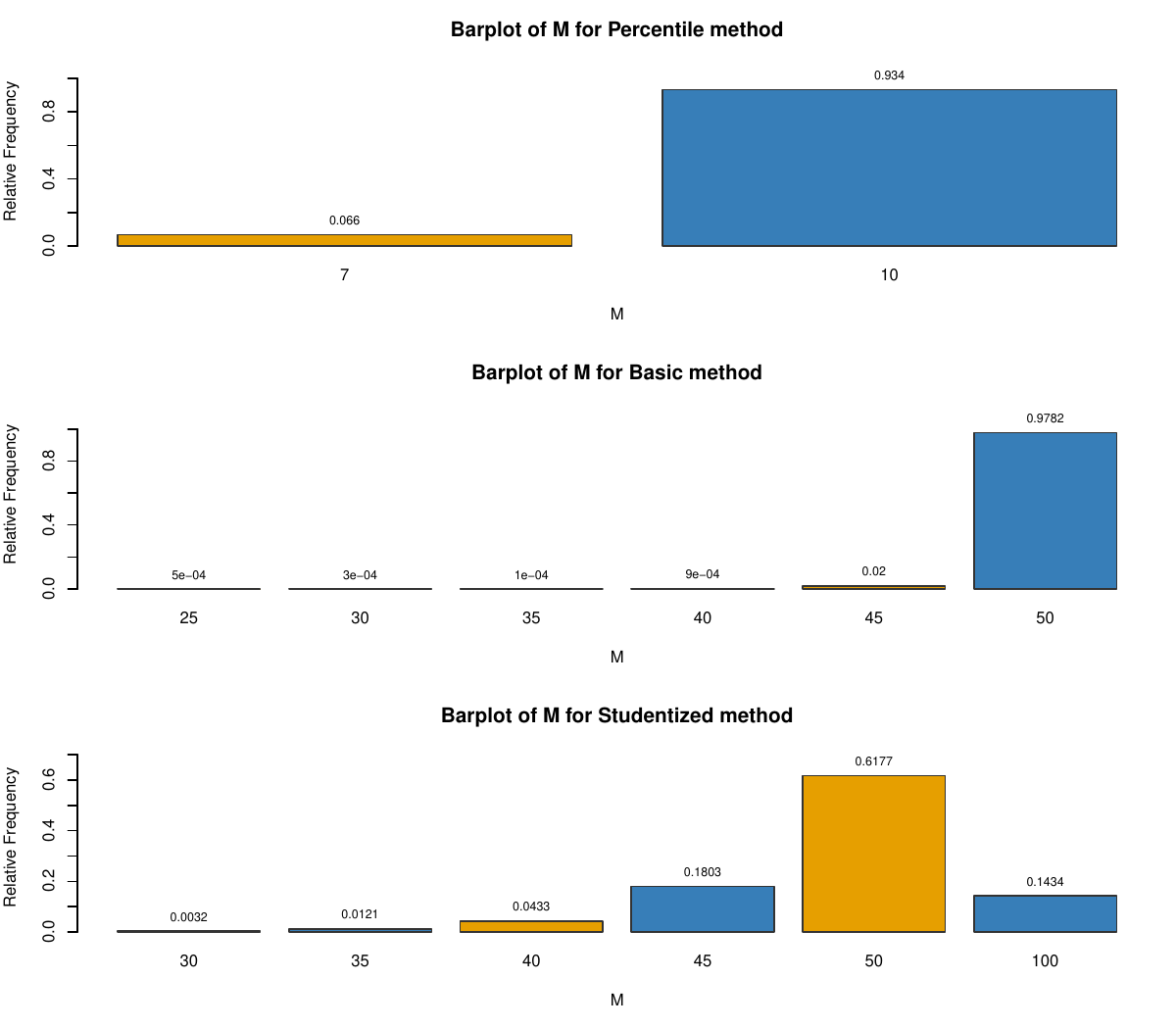}
			\vspace{-0.25cm}
		\end{subfigure}
		\vspace{-0.65cm}
		\caption{Comparison of the chosen number of atoms $\widehat{M}$ for $Q$ in different methods when constructing confidence intervals for the Mean.}
		\label{fig:all_methods_M}
		\vspace{-0.25cm}
	\end{figure}
	
Finally, in Figure \ref{fig:all_methods_M} we have a look at the distribution of the data-adaptively chosen number of support points $\widehat{M}$ for the source distribution $Q$. Here we see very nicely why the OT-based intervals with percentile bootstrap outperforms the others; the regularization effect is much stronger as you only chose between two values and both are small.

	\subsection{Additional Simulation Results for the Correlation Coefficient}\label{subsec:correlation}

Again, like for the simulations when inferring on the mean, we computed the median length and the left and right error probabilities of the confidence intervals, summarized in Table \ref{tab:cor-length-miss}. Here is no clear symmetry problem, and consequently less insight gained from these figures. As expected, the better coverage of OT-intervals is paid by longer intervals. 

\begin{table}[htb]
	\centering  
	\caption{Comparison of length and tail errors of the confidence intervals}
	\label{tab:cor-length-miss}  \vspace{-0.25cm}
	\begin{tabular}{lcccc}
		\toprule
		& \multicolumn{1}{c}{\textbf{OT}} & \multicolumn{1}{c}{\textbf{Bootstrap}}
		& \multicolumn{1}{c}{\textbf{OT}} & \multicolumn{1}{c}{\textbf{Bootstrap}} \\
		\cmidrule(lr){2-2} \cmidrule(lr){3-3} 	\cmidrule(lr){4-4} \cmidrule(lr){5-5}
		$n$ & \multicolumn{2}{c}{Median Length} & \multicolumn{2}{c}{Tail Errors} \\
		\midrule
		5 & 1.7209 & 1.6234 & (0.0533,0.0429) & (0.0719,0.0588) \\
		10 & 1.0803 & 0.9931 & (0.0483, 0.0485) & (0.0862, 0.0698) \\
		25 & 0.7475 & 0.6270 & (0.0426, 0.0315) & (0.0674, 0.0590) \\
		50 & 0.4756 & 0.4499 & (0.0498, 0.0485) & (0.0607, 0.0593) \\
		100 & 0.3364 & 0.3215 & (0.0466, 0.0512) & (0.0551, 0.0585) \\
		\bottomrule
	\end{tabular}
	\vspace{-0.3cm}
\end{table}

Finally, similar to above we can compare the performance of our data-driven method with OT-based intervals using a source distribution $Q$ with fixed Beta-parameters, see Table \ref{tab:beta_cor}. As expected, the intervals work well but worse than the data-adaptive ones.
		\begin{table}[htb]
		\centering  
		\caption{Comparison of Confidence Interval for fixed Beta parameters}
		\label{tab:beta_cor}
		\setlength{\tabcolsep}{2pt} 
		\resizebox{\textwidth}{!}{
			\begin{tabular}{l S S S S S S S S S}
				& \multicolumn{3}{c}{$(\eta, \zeta) = (2,5)$} 
				& \multicolumn{3}{c}{$(\eta, \zeta) = (5,2)$} 
				& \multicolumn{3}{c}{$(\eta, \zeta) = (5,5)$} \\ 
				\toprule
				n & \multicolumn{1}{c}{Coverage} & \multicolumn{1}{c}{Var} & \multicolumn{1}{c}{MSE} & \multicolumn{1}{c}{Coverage} & \multicolumn{1}{c}{Var} & \multicolumn{1}{c}{MSE} & \multicolumn{1}{c}{Coverage} & \multicolumn{1}{c}{Var} & \multicolumn{1}{c}{MSE}   \\    	\midrule
5   & .8843 & .0012 & .0014 & .8825 & .0012 & .0015 & .8708 & .0014 & .0023 \\
10  & .8951 & .0013 & .0013 & .8882 & .0013 & .0013 & .8765 & .0013 & .0018 \\
25  & .9089 & .0008 & .0008 & .9098 & .0008 & .0009 & .8957 & .0009 & .0009 \\
50  & .8846 & .0011 & .0013 & .8836 & .0012 & .0014 & .8831 & .0011 & .0014 \\
100 & .8904 & .0011 & .0012 & .8873 & .0011 & .0013 & .8903 & .0012 & .0012 \\
				\bottomrule
			\end{tabular} }
		\end{table}

\end{document}